\documentclass[12pt]{article}
\usepackage{times,epsfig,graphics,amssymb,latexsym,amsmath,setspace,fullpage,algorithm,algorithmic,subcaption}
\usepackage{array,url}
\usepackage[usenames]{color}
\usepackage{enumitem}
\usepackage{multirow}
\usepackage[comma, numbers]{natbib}

\usepackage{bm} 
\usepackage{appendix,comment}
\usepackage{amsthm}
\usepackage{makecell}

\newtheorem{theorem}{Theorem}

\newtheorem{lemma}{Lemma} 
\newtheorem{proposition}{Proposition}

\newtheorem{remark}{Remark}

\newcommand{\PA}{{\mathcal Q}}
\newcommand{\OS}{{\mathcal O}}

\newcommand\indep{\protect\mathpalette{\protect\independenT}{\perp}}
\def\independenT#1#2{\mathrel{\rlap{$#1#2$}\mkern2mu{#1#2}}}

\def\argmin{\mathop{\rm argmin}}

\def\Cov{\mathop{\rm Cov}}
\def\Var{\mathop{\rm Var}}
\def\sign{\mathop{\rm sign}}
\def\rank{\mathop{\rm rank}}

\def\diag{\mathop{\rm diag}}

\newtheorem{Assumption}{Assumption}

\newcommand {\bfSigma} {\mbox{\boldmath $\Sigma$}}

\newcommand {\bfOmega} {\mbox{\boldmath $\Omega$}}

\def\bx{\mathop{\bf x}}
\def\bX{\mathop{\bf X}}

\newcommand{\eee}{\boldsymbol \epsilon}

\newcommand{\OOO}{\boldsymbol \Omega}
\newcommand{\SG}{\boldsymbol \Sigma}

\newcommand{\TTT}{\boldsymbol \Theta}

\newcommand{\A}{\mathbf A}
\newcommand{\AAA}{\mathcal A}

\newcommand{\BB}{\mathbf B}
\newcommand{\BBB}{\mathcal B}

\newcommand{\OO}{\mbox{$\mathbf O$}}

\newcommand{\CC}{\mathbf C}
\newcommand{\CCC}{\mathcal C}

\newcommand{\DD}{\mathbf D}
\newcommand{\DDD}{\mathbf \Delta}
\newcommand{\DDDD}{\mathcal D}

\newcommand{\ee}{\mathbf e}
\newcommand{\EEE}{\mathbb E}

\newcommand{\MM}{\mathbf M}
\newcommand{\MMM}{\mathcal M}
\newcommand{\YY}{\mbox{$\mathbf Y$}}

\newcommand{\NNN}{\mathcal N}

\newcommand{\II}{\mathbf I}
\newcommand{\III}{\mathcal I}

\newcommand{\LL}{\mathbf L}

\newcommand{\PP}{\mathbf P}
\newcommand{\PPP}{\mathcal P}

\newcommand{\R}{\mathbb R}

\newcommand{\SSS}{\mathbf S}
\newcommand{\SSSS}{\mathcal S}

\newcommand{\TT}{\mathcal T}

\newcommand{\UU}{\mathbf U}

\newcommand{\VV}{\mathbf V}

\newcommand{\xx}{\mathbf x}
\newcommand{\XX}{\mathbf X}

\newcommand{\pa}{\text{pa}}

\newcommand{\1}{\uppercase\expandafter{\romannumeral1}}
\newcommand{\2}{\uppercase\expandafter{\romannumeral2}}

\newcommand{\off}{\text{off}}

\newcommand{\0}{\textbf{0}}

\allowdisplaybreaks[4]

\begin{document}
	
\title{Identifiability and Consistent Estimation for Gaussian Chain Graph Models}
\author{Ruixuan Zhao$^{\dag}$, Haoran Zhang$^\ddag$ and Junhui Wang$^*$\\ [10pt]
	$^\dag$School of Data Science \\
	City University of Hong Kong 
	\and
	$^\ddag$ Department of Statistics and Data Science \\
	Southern University of Science and Technology 
	\and
	$^*$Department of Statistics \\
	The Chinese University of Hong Kong
}
\date{ }

\maketitle
	
\onehalfspacing
\begin{abstract}
	The chain graph model admits both undirected and directed edges in one graph, where symmetric conditional dependencies are encoded via undirected edges and asymmetric causal relations are encoded via directed edges. Though frequently encountered in practice, the chain graph model has been largely under investigated in the literature, possibly due to the lack of identifiability conditions between undirected and directed edges. In this paper, we first establish a set of novel identifiability conditions for the Gaussian chain graph model, exploiting a low rank plus sparse decomposition of the precision matrix. Further, an efficient learning algorithm is built upon the identifiability conditions to fully recover the chain graph structure. Theoretical analysis on the proposed method is conducted, assuring its asymptotic consistency in recovering the exact chain graph structure. The advantage of the proposed method is also supported by numerical experiments on both simulated examples and a real application on the Standard \& Poor 500 index data.
\end{abstract}


\doublespacing

\section{Introduction}	\label{sec:intro}

Graphical models have attracted tremendous attention in recent years, which provide an efficient modeling framework to characterize various relationships among multiple objects of interest. They find applications in a wide spectrum of scientific domains, ranging from finance\cite{Sanford2012}, information system \cite{Stanton1995}, genetics \cite{Friedman2008}, neuroscience \cite{Cole2013} to public health \cite{Luke2007}.

In the literature, two types of graphical models have been extensively studied. The first type is the undirected graphical model, which encodes conditional dependences among collected nodes via undirected edges. Various learning methods have been proposed to reconstruct the undirected graph, especially under the well-known Gaussian graphical model \cite{Friedman2008, Cai2011} where conditional dependences are encoded via the zero pattern of the precision matrix. Another well-studied graphical model is the directed acyclic graphical model, which uses directed edges to represent causal relationships among collected nodes in a directed acyclic graph (DAG). To reconstruct the DAG structure, linear Gaussian structural equation model (SEM) has been popularly considered in the literature where causal relations are encoded via the sign pattern of the coefficient matrix. Various identifiability conditions \cite{Peters2014, Park2020} have been established for the linear Gaussian SEM model, leading to a number of DAG learning methods \cite{ChenW2019, Park2020}.

Another more flexible graphical model, known as the chain graph model, can be traced back to the early work in \cite{Lauritzen1989, Wermuth1990}. It admits both undirected and directed edges in one graph, where symmetric conditional dependencies are encoded via undirected edges and asymmetric causal relations are encoded via directed edges. Further, it is often assumed that no semi-directed cycles are allowed in the chain graph. As a direct consequence, the chain graph model can be seen as a special DAG model with multiple chain components, where each chain component is a subset of nodes connected via undirected edges, and directed edges are only allowed across different chain components. It has been frequently encountered in various application domains, ranging from genetics \cite{Ha2021}, social science \citep{ogburn2020causal} to computer science \cite{Chen2018}. For instance, in pan-cancer network analysis \citep{Ha2021}, there are three kinds of variables: DNA-level variables containing copy number and methylation, transcriptomic variables containing mRNA expression and proteomic variables. It is believed that there only exist symmetric dependencies within three kinds of variables and asymmetric casual relationships across different kinds of variables. Also, in the stock market, the relationship among different stocks can be multi-faceted, where there are symmetric competitive or cooperative relationship and  asymmetric causal relationship between pairs of stocks. 

The chain graph model can have different interpretations, including the Andersson-Madigan-Perlman (AMP) interpretation \cite{Andersson2001}, the Lauritzen–Wermuth–Frydenberg (LWF) interpretation \cite{Lauritzen1989, Frydenberg1990} and the multivariate regression (MVR) interpretation \cite{Cox1993}. Each interpretation implies a different  conditional independence relationship from the chain graph structure. Particularly, the LWF interpretation leads to a natural interpretation \citep{Cox1993, Andersson2001}, where each node affects its children and the nodes in the same chain components of these children \cite{Javidian2020}. The AMP interpretation, however, is more suitable for statistical modeling via the SEM model, which also facilitates the data generating scheme of the chain graph \citep{Andersson2001, Drton2006}. In contrast to the LWF chain graph models, each node in AMP chain graph models only affects its children, but no other nodes in the same chain components of these children \cite{Javidian2020}. The MVR interpretation corresponds to the acyclic directed mixed graph (ADMG) with no partially directed cycle \cite{Javidian2018properties}. 

Based on different interpretations, several structure learning methods are established, including the IC like algorithm \cite{Studeny1997}, the CKES algorithm \cite{Pena2014} and the decomposition-based algorithm \cite{Ma2008}, the Markov blanket based algorithm \cite{Javidian2020b}, the pseudolikelihood-type algorithm \cite{Bhattacharya2020} for LWF chain graphs, the PC like algorithm \cite{Sonntag2012} and the decomposition-based algorithm \cite{Javidian2018} for MVR chain graphs, and the PC like algorithm \cite{Pena20142} and the decomposition-based algorithm \cite{Javidian2020} for AMP chain graphs. Yet, all these methods can only estimate some Markov equivalence classes of the chain graph model, and provide no guarantee for the reconstruction of the exact chain graph structure, mostly due to the lack of identifiability conditions between undirected and directed edges. It was until very recently that \cite{Wang2021} extends the equal noise variance assumption for DAG \cite{Peters2014} to establish the identifiability of the chain graph model under the AMP interpretation, which is rather difficult to verify in practice. It is also worth mentioning that if the chain components and their causal ordering are known a priori, then the chain graph model degenerates to a sequence of multivariate regression models, and various methods \cite{Drton2006, Mccarter2014, Ha2021} have been developed to recover the graphical structure.


In this paper, we establish a set of novel identifiability conditions for the Gaussian chain graph model under AMP interpretation, exploiting a low rank plus sparse decomposition of the precision matrix. Further, an efficient learning algorithm is developed to recover the exact chain graph structure, including both undirected and directed edges. Specifically, we first reconstruct the undirected edges by estimating the precision matrix of the noise vector through a regularized likelihood optimization.  Then, we identify each chain component and determine its causal ordering based on the conditional variances of its nodes. Finally, the directed edges are reconstructed via multivariate regression coupled with truncated singular value decomposition (SVD). Theoretical analysis shows that the proposed method can consistently reconstruct the exact chain graph structure, which, to the best of our knowledge, is the first asymptotic consistency result in terms of exact graph recovery for the Gaussian chain graph with AMP interpretation and linear SEM model. The advantage of the proposed method is supported by numerical experiments on both simulated examples and a real application on the Standard \& Poor 500 index data, which reveals some interesting impacts of the COVID-19 pandemic on the stock market.

The rest of the paper is organized as following. Section \ref{sec:model} introduces some preliminaries on the chain graph model. Section \ref{sec::method} proposes the identifiability conditions for linear Gaussian chain graph model, and develops an efficient learning algorithm to reconstruct the exact chain graph structure. The asymptotic consistency of the proposed method is established in Section \ref{sec:theory}. Numerical experiments of the proposed method on both simulated and real examples are included in Section \ref{sec:experiment}. Section \ref{sec:conclusion} concludes the paper, and technical proofs are provided in the Appendix. Auxiliary lemmas and further computational details are deferred to a separate Supplementary File.

Before moving to Section \ref{sec:model}, we define some notations. For an integer $m$, denote $[m] = \{1,...,m\}$.
For a real value $x$, denote $\lceil x\rceil$ as the largest integer less than or equal to $x$.
For two nonnegative sequences $a_n$ and $b_n,~a_n\lesssim b_n$ means there exists a constant $c>0$ such that $a_n\leq cb_n$ when $n$ is sufficiently large. Further, $a_n\lesssim_P b_n$ means there exists a constant $c>0$ such that $\Pr(a_n\leq cb_n) \to 1$ as $n$ grows to infinity. 
For a vector $\xx$, the sub-vector corresponding to an index subset $S$ is denoted as $\xx_S=(\xx_i)_{i\in S}$. For a matrix $\mathbf{A}=(a_{ij})_{p\times p}$, the sub-matrix corresponding to rows in $S_1$ and columns in $S_2$ is denoted as $\mathbf{A}_{S_1,S_2}=(a_{ij})_{i\in S_1,j\in S_2}$, and let $\mathbf{A}_{S_1,S_2}^{-1}$ denote the corresponding sub-matrix of $\A^{-1}.$  
Also, let $\|\mathbf{A}\|_{1,\text{off}}=\sum_{i\neq j} |a_{ij}|$, $\|\A\|_{\max} = \max_{ij} |a_{ij}|$, $\|\mathbf{A}\|_2$ denote the spectrum norm, $\|\mathbf{A} \|_*$ denote the nuclear norm and $v(\A)\in\R^{p^2}$ denote the vectorization of $\A$.

\section{Chain graph model}\label{sec:model}

Suppose the joint distribution of $\xx=(x_1,...,x_p)^\top$ can be depicted as a chain graph ${\cal G}=({\cal N}, {\cal E})$, where ${\cal N}=\{1,...,p\}$ represents the node set and ${\cal E} \subset {\cal N} \times {\cal N}$ represents the edge set containing all undirected and directed edges. To differentiate, we denote $(i - j)$ for an undirected edge between nodes $i$ and $j$, $(i \rightarrow j)$ for a directed edge pointing from node $i$ to node $j$, and suppose that at most one edge is allowed between two nodes.
Then, there exists a positive integer $m$ such that ${\cal N}$ can be uniquely partitioned into $m$ disjoint chain components $\NNN = \bigcup_{k=1}^m \tau_k,$ where each $\tau_k$ is a connected component of nodes via undirected edges.
Suppose that only undirected edges exist within each chain component, and directed edges are only allowed across different chain components \cite{Maathuis2018, Drton2006}. Further suppose that there exists a permutation $\boldsymbol{\pi}=(\pi_1,...,\pi_m)$ such that for $i \in \tau_{\pi_k}$ and $j \in \tau_{\pi_l}$, if $(i\rightarrow j)\in{\cal E}$, then $k<l$. 
This excludes the existence of semi-directed cycle in $\cal G$ \cite{Maathuis2018}.
We call such permutation $\boldsymbol{\pi}$ as the causal ordering of the chain components, and directed edges could only point from nodes in higher-order $\tau_k$ to nodes in lower-order one.

Let $\pa(i) = \{ j\in {\cal N} : (j\rightarrow i) \in {\cal E}\},~\mbox{ch}(i)=\{j\in {\cal N} : (i\rightarrow j) \in {\cal E}\}$ and $\mbox{ne}(i) = \{ j\in {\cal N} :  (j-i) \in {\cal E}\}$ denote the parents, children and neighbors of node $i,$ respectively. Further, let $\pa(\tau_k)=\bigcup_{i \in \tau_k} \pa(i)$ be the parent set of chain component $\tau_k$. Suppose the joint distribution of $\xx$ satisfies the AMP Markov property \citep{Andersson2001, Levitz2001} with respect to ${\cal G},$ and follows the linear structural equation model (SEM),
\begin{align}\label{eq:model}
	\bx = \mathbf{B} \bx + \eee,
\end{align}
where $\mathbf{B} = (\beta_{ij})_{p\times p}$ is the coefficient matrix, $\eee = (\epsilon_1,...,\epsilon_p)^\top \sim N(\mathbf{0}, \mathbf{\Omega}^{-1})$, and $\mathbf{\Omega}=(\omega_{ij})_{p\times p}$ is the precision matrix of $\eee$. 
Further, suppose that $\beta_{ij} \neq 0$ if and only if $j\in \pa(i)$, and $\omega_{ij}\neq 0$ if and only if $j\in \mbox{ne}(i)$.
Therefore, the undirected and directed edges in $\cal G$ can be directly implied by the zero patterns in $\OOO$ and $\BB$, respectively. 
The joint density of $\xx$ can then be factorized as 
\begin{equation}\label{eq:factorization}
	P(\xx)=\prod_{k=1}^m P(\xx_{\tau_k} | \xx_{\pa(\tau_k)}),
\end{equation}
where $\mathbf{x}_{\tau_k} | \mathbf{x}_{\pa(\tau_k)} \sim N(\mathbf{B}_{\tau_k, \pa(\tau_k)} \mathbf{x}_{\pa(\tau_k)}, \mathbf{\Omega}_{\tau_k,\tau_k}^{-1})$ for $k\in[m]$, and  $\mathbf{\Omega}_{\tau_k,\tau_k}$ is not necessarily a diagonal matrix. This is a key component of the chain graph model to allow undirected edges within each $\tau_k$, and thus differs from most existing SEM models with diagonal $\mathbf{\Omega}$ in the literature \cite{Peters2017, Peters2014, Park2020, ChenW2019}. 

It is also interesting to compare the AMP, LWF and MVR interpretations of $\cal G$ with the SEM model in \eqref{eq:model}. Let $\TTT^k$ be the precision matrix corresponding to $\tau_k\cup \pa(\tau_k)$, then different interpretations lead to different implications on $\OOO$ and $\BB$. Particularly, 
\begin{itemize}
    \item AMP interpretation: $(j\rightarrow i)\notin {\cal E}$ for $i\in \tau_k$ and $j\in \pa(\tau_k)$ $\implies \BB_{ij} = 0$ ; $(j - i) \notin {\cal E}$ for $i,j \in \tau_k$ $\implies \OOO_{ij} = 0$.
	\item LWF interpretation: $(j\rightarrow i)\notin {\cal E}$ for $i\in \tau_k$ and $j\in \pa(\tau_k)$ $\implies [\TTT^k_{\tau_k,\pa(\tau_k)}]_{ij} = 0$; $(j - i) \notin {\cal E}$ for $i,j \in \tau_k$ $\implies \OOO_{ij} = 0$.
	\item MVR interpretation: $(j\rightarrow i)\notin {\cal E}$ for $i\in \tau_k$ and $j\in \pa(\tau_k)$ $\implies \BB_{ij} = 0$; $(j \leftrightarrow i) \notin {\cal E}$ for $i,j \in \tau_k$ $\implies \OOO_{ij}^{-1} = 0$.
\end{itemize}
It is clear that AMP interpretation matches well with the linear SEM model so that the graph structure in $\cal G$ can be fully captured by the supports of $\OOO$ and $\BB$, which also greatly facilitates the data generating scheme of the chain graph \citep{Andersson2001, Drton2006}. Yet, the directed edges in the LWF chain graph corresponds to more involved parameter $\TTT^k_{\tau_k,\pa(\tau_k)} = - \OOO_{\tau_k,\tau_k} \BB_{\tau_k,\pa(\tau_k)}$, whereas the undirected edges in the MVR chain graph corresponds to the noise covariance matrix $\OOO^{-1}$. Here $[\TTT^k_{\tau_k,\pa(\tau_k)}]_{ij}$ represents the element in matrix $\TTT^k_{\tau_k,\pa(\tau_k)}$ corresponding to nodes $i$ and $j$.


Furthermore, to assure the acyclicity among chain components in $\cal G$, we say $(\OOO,\BB)$ is CG-feasible if there exists a permutation matrix $\PP$ such that both $\PP\OOO\PP^\top$ and $\PP\BB\PP^\top$ share the same block structure, where $\PP\OOO\PP^\top$ is a block diagonal matrix and $\PP\BB\PP^\top$ is a block lower triangular matrix with zero diagonal blocks. Figure \ref{ToyExample} shows a toy chain graph, as well as the supports of the original and permuted $(\OOO,\BB)$. Let $\TTT$ denote the precision matrix for $\xx$, then it follows from \eqref{eq:model} that 
\begin{align} \label{eq::decom}
	\mathbf{\Theta}=(\mathbf{I}_p-\mathbf{B})^\top \mathbf{\Omega} (\mathbf{I}_p-\mathbf{B}) =: \mathbf{\Omega} + \mathbf{L},
\end{align}
where $\mathbf{L} = \mathbf{B}^\top \mathbf{\Omega}\mathbf{B} - \mathbf{B}^{\top} \mathbf{\Omega} - \mathbf{\Omega}\mathbf{B}$. 


\begin{figure}[!htb]
	\centering
	\begin{minipage}[b]{0.45\textwidth}
		\centering
		\includegraphics[width=1\textwidth]{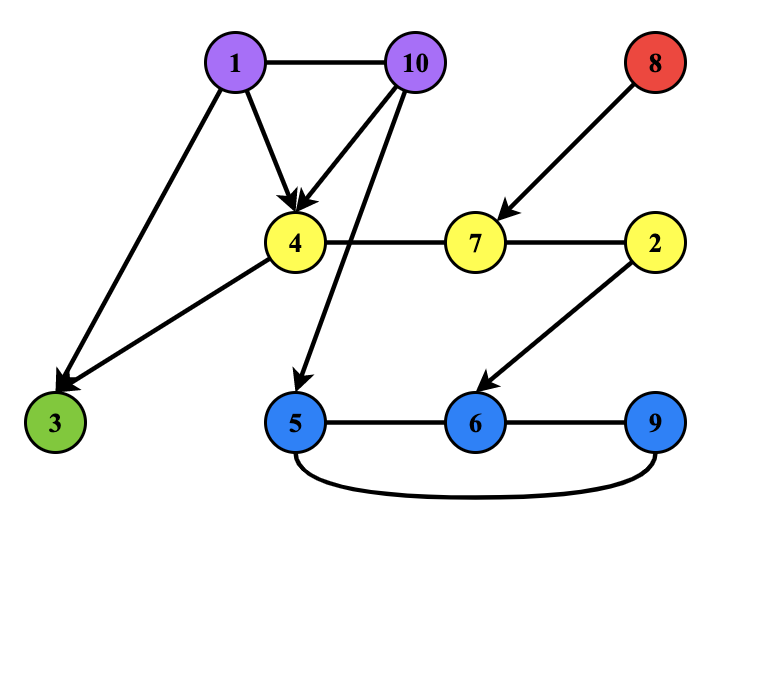}
	\end{minipage}
	\begin{minipage}[b]{0.5\textwidth}
		\centering
		\includegraphics[width=1\textwidth]{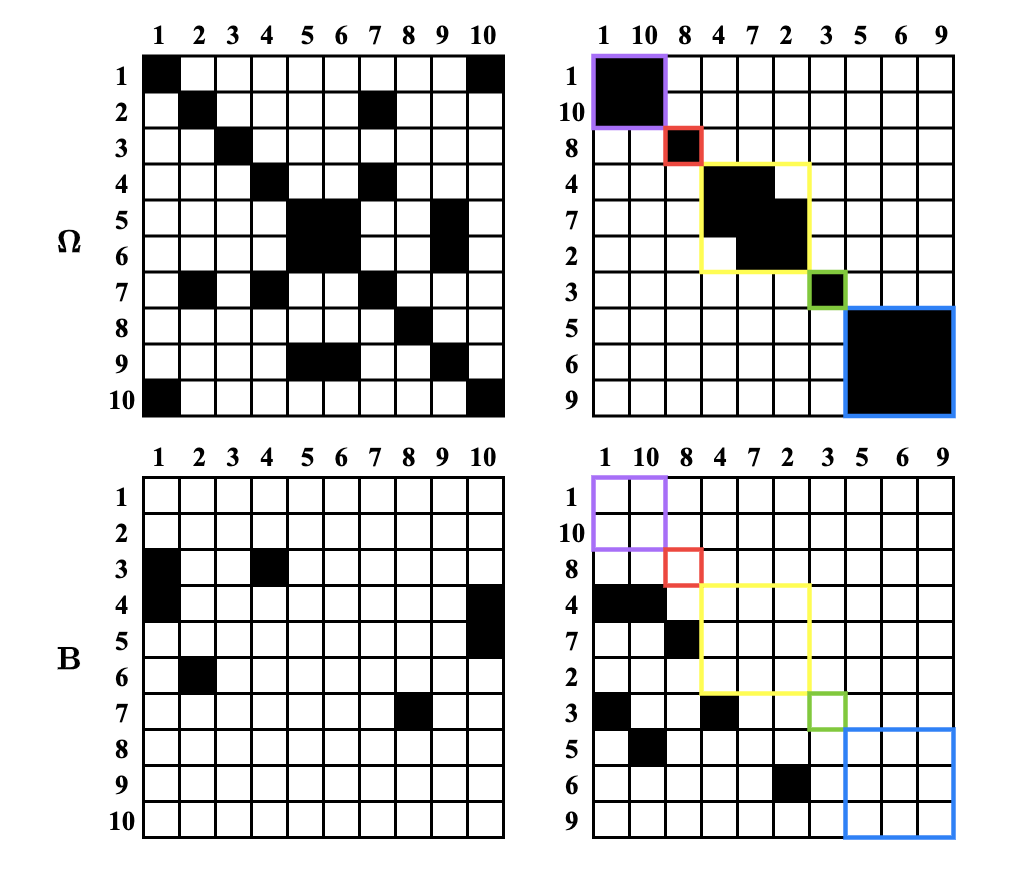}
	\end{minipage}
	\caption{The left panel displays a toy chain graph with colors indicating different chain components, and the right panel displays the supports of the original $(\OOO,\BB)$ in the first column and the permuted $(\OOO,\BB)$ in the second column.}\label{ToyExample}
\end{figure}

\section{Proposed method}\label{sec::method}

\subsection{Identifiability of $\cal G$}\label{sec::identi}

A key challenge in the chain graph model is the identifiability of the graph structure $\cal G$, due to the fact that $\OOO$ and $\BB$ are intertwined in the SEM model in \eqref{eq:model}. To proceed, we assume that $\OOO$ is sparse with $\|\OOO\|_0 = S$ and $\LL$ is a low-rank matrix with $\rank(\LL) = K$. The sparseness of $\OOO$ implies the sparseness of undirected edges in ${\cal G}$, which has been well adopted in the literature of Gaussian  graphical models \cite{Meinshausen2006, Friedman2008, Cai2011}. The low-rank of $\LL$ inherits from that of $\BB$ \cite{Fang2023}, which essentially assumes the presence of hub nodes in $\cal G$ or nodes with multiple children or parents.

Let $\LL = \UU_1 \DD_1 \UU_1^\top$ be the eigen decomposition of $\LL,$ where $\UU_1^\top \UU_1 = \II_K$ and $\DD_1$ is a $K\times K$ diagonal matrix. We define two linear subspaces, 
\begin{eqnarray}
	{\cal S}(\OOO) &=& \{\SSS\in\R^{p\times p}: \SSS^\top = \SSS,~\text{and}~s_{ij} = 0~\text{if}~\omega_{ij} = 0\}, \label{eq:tangent S} \\
	{\cal T}(\LL) &=& \left\{ \UU_1\YY+\YY^\top\UU_1^\top:~ \YY \in \R^{K\times p} \right\}, \label{eq:tangent L}
\end{eqnarray}
where ${\cal S}(\OOO)$ is the tangent space, at point $\OOO$, of the manifold containing symmetric matrices with at most $S$ non-zero entries, and ${\cal T}(\LL)$ 
is the tangent space, at point $\LL$, of the manifold containing symmetric matrices with rank at most $K$ \citep{Chandrasekaran2011}. 

\begin{Assumption}\label{ass:tangent}
	${\cal S}(\OOO)$ and ${\cal T}(\LL)$ intersect at the origin only; that is, ${\cal S}(\OOO) \cap {\cal T}(\LL) = \{\mathbf{0}_{p\times p}\}$.	
\end{Assumption}

Assumption \ref{ass:tangent} is the same as the transversality condition in \cite{Chandrasekaran2011}, which assures the identifiability of $(\OOO,\LL)$ in the sense that $\TTT$ can be uniquely decomposed as the sum of a matrix in ${\cal S}(\OOO)$ and the other one in ${\cal T}(\LL)$. It essentially holds if $\OOO$ is sparse but not low-rank and $\LL$ is low-rank but not sparse. Note that $\OOO$ is the precision matrix and thus full-rank, yet its sparsity corresponds to the sparseness of undirected edges in ${\cal G}$. The low-rank of $\LL$ inherits from that of $\BB$, which corresponds to the presence of hub nodes in ${\cal G}$ \citep{Fang2023}. Further, we have
$$L_{jk} = \sum_{l=1}^p\sum_{i=1}^p \beta_{ij}\omega_{il}\beta_{lk} - \sum_{i=1}^p \beta_{ij}\omega_{ik} - \sum_{i=1}^p \omega_{ji}\beta_{ik},$$
where the first term corresponds to the path $j\rightarrow i - l \leftarrow k$ if $l\neq i$ or the path $j\rightarrow i \leftarrow k$ if $l=i$, the second term corresponds to the path $j\rightarrow i - k$ if $i\neq k$ and $j\rightarrow k$ if $i=k$, and the third term corresponds to the path $j-i\leftarrow k$ if $i\neq j$ and $j\leftarrow k$ if $i=j$. Therefore, $L_{jk}\neq 0$ if any one of these paths exists between nodes $j$ and $k$, and thus it is reasonable to assume $\LL$ to be non-sparse. 

\begin{Assumption}\label{ass:diff}
	The $K$ eigenvalues of $\LL$ are distinct.
\end{Assumption}

Assumption~\ref{ass:diff} is necessary to identify the eigen space of the low-rank matrix $\LL,$ which has been commonly assumed in the literature of matrix perturbation \cite{yu2015useful}. 
Let $\PA$ be the parameter space of CG-feasible $(\OOO,\BB)$, where $\OOO\succ0,~\OOO+\LL\succ0,~\|\OOO\|_0\leq S,~\rank(\LL)\leq K,$ and Assumptions~\ref{ass:tangent} and \ref{ass:diff} are met. Let $(\OOO^*,\BB^*)$ denote the true parameters of the linear SEM model \eqref{eq:model} satisfying $\|\OOO^*\|_0 = S$ and $\rank(\LL^*)=K.$

\begin{theorem}\label{thm:ident}
	Suppose $(\OOO^*,\BB^*)\in\PA.$ Then, there exists a small $\epsilon>0$ such that for any $(\OOO, \BB)\in\PA$ satisfying $\| \OOO-\OOO^*\|_{\max}<\epsilon$ and $\|\BB-\BB^*\|_{\max}<\epsilon,$ if $$
	(\mathbf{I}_p-\BB)^\top\OOO (\mathbf{I}_p-\BB)=(\mathbf{I}_p-\BB^*)^\top \OOO^* (\mathbf{I}_p-\BB^*),
	$$ 
	it holds true that $(\OOO, \BB) = (\OOO^*,\BB^*).$
\end{theorem}

Theorem \ref{thm:ident} establishes the local identifiability of $(\OOO^*,\BB^*)$ in \eqref{eq:model}, which further implies the local identifiability of the chain graph $\cal G$. It essentially states that $(\OOO^*,\BB^*)\in\PA$ can be uniquely determined, within its neighborhood in $\PA$, by the precision matrix $\TTT^*=(\mathbf{I}_p-\BB^*)^\top \OOO^* (\mathbf{I}_p-\BB^*),$ which can be consistently estimated from the observed sample. 

\begin{remark}[{\bf Identifiability for DAG}]
	When $\OOO$ is indeed a diagonal matrix, each chain component contains exactly one node and thus $\cal G$ reduces to a DAG. By Theorem \ref{thm:ident}, it is identifiable as long as the eigenvalues of $\LL^*$ are distinct and $\ee_j\notin\text{span}(\UU_1^*)$ for any $j \in [p]$, where $\UU_1^*\in\R^{p\times K}$ contains eigenvectors of $\LL^*$ and $\{\ee_j\}_{j=1}^p$ are the standard basis of $\R^p.$ It provides an alternative identifiability condition for DAG, in contrast to the popularly-employed equal noise variance assumption \cite{Peters2014}. These two assumptions are incomparable though. Particularly, in our identifiability condition for DAG, $\OOO$ is a diagonal matrix but the diagonal elements can vary. On the other hand, the equal noise variance assumption also does not imply our identifiability condition. For example, let $p=3$, $\OOO = \mathbf{I}_3$ and $\BB = (b_{ij})_{3\times 3}$ with $b_{21}=1$ and other entries being 0, then
	$$
	\LL = 
	\begin{pmatrix}
		1 & -1 & 0\\
		-1 & 0 & 0 \\
		0 & 0 & 0
	\end{pmatrix}
	$$ where it is easy to verify that $\ee_1, \ee_2 \in \text{span}(\UU_1)$.
\end{remark}


\subsection{Learning algorithm}\label{sec:est}

We now develop a learning algorithm to estimate $(\mathbf{\Omega}^*, \mathbf{B}^*)$ and reconstruct the chain graph ${\cal G}$. Suppose we observe independent copies $\xx_1,...,\xx_n\in\R^p$ and denote $\mathbf{X}=(\mathbf{x}_1,..., \mathbf{x}_n)^\top\in \mathbb{R}^{n\times p}$.
We first estimate $\mathbf{\Omega}^*$ via the following regularized likelihood,
\begin{align} \label{eq:opti1}
	(\widehat\OOO,\widehat\LL) = \argmin_{\OOO,\LL} & \ \big \{-l(\OOO+\LL) + \lambda_n  (\|\OOO\|_{1,\off} + \gamma\|\LL\|_* ) \big \} \\
	\mbox{subject to} &~~ \OOO \succ 0 \mbox{~~and~~}\OOO+\LL \succ0, \nonumber
\end{align}
where $l(\TTT) = -\text{tr}(\mathbf{\Theta} \widehat \SG) + \log\{\det (\mathbf{\Theta})\}$ is the Gaussian log-likelihood with $\widehat\SG=\frac{1}{n}\mathbf{X}^\top \mathbf{X}$, and $\lambda_n$ and $\gamma$ are tuning parameters. Here $\|\OOO\|_{1,\off} = \sum_{i\neq j}|\omega_{ij}|$ induces sparsity in $\OOO$, $\|\LL\|_*$ induces low-rank of $\LL$, and the constraints are due to the fact that both $\OOO$ and $\mathbf{\Theta}=\OOO+\LL$ are precision matrices. Note that the optimization task in \eqref{eq:opti1} is convex, and can be efficiently solved via the alternative direction method of multipliers (ADMM; \cite{BoydS2011}). More computational details are deferred to the Supplementary Files. 

Once $\widehat\OOO$ is obtained, we connect nodes $i$ and $j$ with undirected edges if $\widehat\omega_{ij}\neq 0$, which leads to multiple estimated chain components, denoted as $\widehat\tau_1,...,\widehat\tau_{\widehat m}$. We can further determine their causal ordering by comparing the estimated conditional variance of node $i\in \widehat\tau_k$ with $\widehat{\OOO}_{ii}^{-1}$. Particularly, the conditional variance of node $i$ can be decomposed into two parts, with one corresponding to its parent nodes and the other from $\OOO_{ii}^{-1}$. Therefore, if all the parent nodes of node $i$ are given, the remaining variance should all come from $\OOO_{ii}^{-1}$, otherwise the remaining variance should be larger than $\OOO_{ii}^{-1}$. This observation shares the same spirit as \cite{ChenW2019} in determining causal ordering, without requiring equal noise variances though. More specifically, for each $\widehat\tau_k$ and any $\CCC\subset[p]\backslash\widehat\tau_k$, we define 
$$
\widehat{\DDDD}(\widehat{\tau}_k,\CCC)= \max_{i\in \widehat{\tau}_k} \Big \{ \widehat{\mathbf{\Sigma}}_{ii} - \widehat{\mathbf{\Sigma}}_{i {\CCC}} \widehat{\mathbf{\Sigma}}_{{\CCC}{\CCC}}^{-1} \widehat{\mathbf{\Sigma}}_{ {\CCC}i} - \widehat{\OOO}_{ii}^{-1} \Big \},
$$
where $\widehat{\mathbf{\Sigma}}_{ii} - \widehat{\mathbf{\Sigma}}_{i {\CCC}} \widehat{\mathbf{\Sigma}}_{{\CCC}{\CCC}}^{-1} \widehat{\mathbf{\Sigma}}_{ {\CCC}i}$ is the estimated conditional variance for node $i\in \widehat\tau_k$ given nodes in $\CCC$, and $\widehat{\OOO}_{ii}^{-1}$ is the estimated variance for node $i\in\widehat\tau_k$ given its parent chain components. It is thus clear that $\widehat{\DDDD}(\widehat{\tau}_k,\CCC)$ shall be close to 0 if $\CCC$ consists of all upper chain components of $\widehat\tau_k$. We start with $\widehat{\cal D}(\widehat{\tau}_k, \emptyset)= \max_{i\in \widehat{\tau}_k} \big \{ \widehat{\mathbf{\Sigma}}_{ii} - \widehat{\OOO}_{ii}^{-1} \big \}$ for each chain component $\widehat{\tau}_k$, and select the first chain component by $\widehat{\pi}_1 = \argmin_{l\in [\widehat m]}\widehat{\cal D}(\widehat{\tau}_l,\emptyset)$.
Suppose the first $s$ chain components $\widehat{\tau}_{\widehat{\pi}_1},...,\widehat{\tau}_{\widehat{\pi}_s}$ have been selected, let $\widehat{\cal C}_s = \cup_{k=1}^s \widehat{\tau}_{\widehat{\pi}_k}$ and $\widehat{\pi}_{s+1} = \argmin_{l\in [\widehat m] \setminus \cup_{k=1}^s \widehat{\pi}_k}\widehat{\cal D}(\widehat{\tau}_l, \widehat{\cal C}_s)$.
We repeat this procedure until the causal orderings of all $\widehat{\tau}_k$'s are determined, which are denoted as $\widehat{\boldsymbol\pi}=(\widehat{\pi}_1,...,\widehat{\pi}_{\widehat{m}})$.

Finally, we estimate $\BB^*$ based on the  estimated causal ordering $\widehat{\boldsymbol\pi}$ of the chain components. Particularly, $\widehat{\cal C}_{k-1}$ consists of all the chain components in front of $\widehat{\pi}_k$, and then the directed edges, if present, can only point from nodes in $\widehat{\cal C}_{k-1}$ to nodes in $\widehat{\pi}_k$. Thus, we first give an intermediate estimate $\widehat{\BB}^{\text{reg}}$, whose submatrix $\widehat{\BB}^{\text{reg}}_{\widehat{\tau}_{\widehat{\pi}_k}, \widehat{\cal C}_{k-1}}$ is obtained via a multivariate regression of $\xx_{\widehat{\tau}_{\widehat{\pi}_k}}$ on $\xx_{\widehat{\cal C}_{k-1}}$. Given $\widehat{\BB}^{\text{reg}}$, we conduct singular value decomposition (SVD) as $\widehat{\BB}^{\text{reg}}=\widehat{\UU}^{\text{reg}}\widehat{\DD}^{\text{reg}}(\widehat{\VV}^{\text{reg}})^\top$ with $\widehat{\DD}^{\text{reg}}=\diag(\widehat{\sigma}_1^{\text{reg}},...,\widehat{\sigma}_p^{\text{reg}})$, and then truncate the small singular values to obtain $\widehat{\BB}^{\text{svd}}=\widehat{\UU}^{\text{reg}}\widehat{\DD}^{\text{svd}}(\widehat{\VV}^{\text{reg}})^\top$, where   $\widehat{\DD}^{\text{svd}}=\diag(\widehat{\sigma}_1^{\text{svd}},...,\widehat{\sigma}_p^{\text{svd}})$ with $\widehat{\sigma}_j^{\text{svd}}=0$ if $\widehat{\sigma}_j^{\text{reg}}\leq\kappa_n$ and $\widehat{\sigma}_j^{\text{svd}}=\widehat{\sigma}_j^{\text{reg}}$ if $\widehat{\sigma}_j^{\text{reg}}>\kappa_n$, for some pre-specified $\kappa_n>0$. 
The final estimate $\widehat{\BB}=(\widehat{\beta}_{ij})_{p\times p}$ is obtained by truncating the diagonal and upper triangular blocks to 0, and conducting a hard thresholding to the lower triangular blocks with some pre-specified $\nu_n>0$, where $\widehat{\beta}_{ij}=0$ if $|\widehat{\beta}_{ij}^{\text{svd}}|\leq \nu_n$ and $\widehat{\beta}_{ij}=\widehat{\beta}_{ij}^{\text{svd}}$ if $|\widehat{\beta}_{ij}^{\text{svd}}|> \nu_n$. 
The nonzero elements of $\widehat{\BB}$ then lead to the estimated directed edges. 

The proposed learning algorithm is computational efficient, whose computational complexity is roughly of polynomial order of $p$. Specifically, in estimating $\OOO^*$, the optimization \eqref{eq:opti1} is convex and can be solved in polynomial time. The causal ordering of chain components is recovered via iteratively computing $\widehat{\cal D}(\widehat{\tau}_k, {\cal C})$ and searching for the chain component with minimal $\widehat{\cal D}(\widehat{\tau}_k, {\cal C})$. It is clear that the above procedure involves computing conditional variances, and finding maximum and minimum, whose complexity is of polynomial order in $p$. Finally, the complexity of  multivariate regressions coupled with truncated SVD to reconstruct the directed edges is also of polynomial order in $p$.

\section{Asymptotic theory}\label{sec:theory}

This section quantifies the asymptotic behavior of $(\widehat\OOO,\widehat\BB)$, and establishes its consistency for reconstructing the chain graph ${\cal G}^* = ({\cal N}, {\cal E}^*)$.
Let $\TTT^* = (\mathbf{I}_p-\BB^*)^\top \OOO^* (\mathbf{I}_p-\BB^*),$ and the Fisher information matrix takes the form 
$$
\III^* = - \mathbb{E} \Big[\frac{\partial^2 l (\mathbf{\Theta}^*)}{\partial \mathbf{\Theta}^2} \Big]  =  (\mathbf{\Theta}^*)^{-1} \otimes (\mathbf{\Theta}^*)^{-1},
$$
where $\otimes$ denotes the Kronecker product. For a linear subspace $\MMM,$ let ${\cal P}_{{\cal M}}$ denote the projection onto ${\cal M},$ and ${\cal M}^{\perp}$ denote the orthogonal complement of ${\cal M}$. We further define two linear operators $F : {\cal S}(\OOO^*) \times {\cal T}(\LL^*) \rightarrow {\cal S}(\OOO^*) \times {\cal T}(\LL^*)$ and $F^\perp : {\cal S}(\OOO^*) \times {\cal T}(\LL^*) \rightarrow {\cal S}(\OOO^*)^\perp \times {\cal T}(\LL^*)^\perp$ such that 
\begin{align*}
	F(\DDD_{\OOO},\DDD_{\LL}) &= ({\cal P}_{{\cal S}(\OOO^*)} ( \III^*v(\DDD_{\OOO} + \DDD_{\LL})), {\cal P}_{{\cal T}(\LL^*)} ( \III^* v(\DDD_{\OOO} + \DDD_{\LL}))), \\
	F^\perp(\DDD_{\OOO},\DDD_{\LL}) &= ({\cal P}_{{\cal S}(\OOO^*)^{\perp}} ( \III^*v(\DDD_{\OOO} + \DDD_{\LL})), {\cal P}_{{\cal T}(\LL^*)^\perp} ( \III^*v(\DDD_{\OOO} + \DDD_{\LL}))).
\end{align*}
According to Assumption \ref{ass:tangent} and Lemma 1 in the Supplementary Files, $F$ is invertible and thus $F^{-1}$ is well defined. 

Let $g_{\gamma}(\OOO,\LL)=\max\{\|\OOO\|_{\max}, \|\LL\|_2/\gamma\},$ where $\gamma>0$ is the same as that in \eqref{eq:opti1}. Let $\LL^* = \UU_1^*\DD_1^*(\UU_1^*)^\top$ be the eigen decomposition of $\LL^*.$


\begin{Assumption}\label{ass:incor}
	$g_{\gamma}(F^{\perp} F^{-1}(\sign(\OOO^*), \gamma \UU_1^*\sign(\DD_1^*) (\UU_1^*)^\top)) < 1$.
\end{Assumption}

Assumption~\ref{ass:incor} is essential for establishing the selection and rank consistency of $(\widehat\OOO,\widehat\LL)$ through penalties in \eqref{eq:opti1}. For example, in a special case when $\TTT^* = \II_p$ and ${\cal S}(\OOO^*) \perp {\cal T}(\LL^*)$, Assumption~\ref{ass:incor} simplifies to $\max\{\gamma\|\UU_1^*\sign(\DD_1^*) (\UU_1^*)^\top\|_{\max},\|\sign(\OOO^*)\|_2/\gamma\}<1$, implying that $\UU_1^*$ is not sparse and $\sign(\OOO^*)$ is not a low-rank matrix. Similar technical conditions have also been assumed in the literature \cite{Chandrasekaran2011, chen2016fused}.

\begin{theorem}\label{thm:Omega}
	Suppose $(\OOO^*,\BB^*)\in\PA$ and Assumption~\ref{ass:incor} holds. Let $\lambda_n = n^{-1/2+\eta}$ with a sufficiently small positive constant $\eta$. Then, with probability approaching 1, \eqref{eq:opti1} has a unique solution $(\widehat\OOO,\widehat\LL).$ Furthermore, we have
	\begin{equation}\label{eq:consis}
		\begin{aligned}
			&\|\widehat\OOO - \OOO^*\|_{\max}  \lesssim_P n^{-\frac{1}{2}+2\eta},~~~~ &&\Pr(\sign(\widehat\OOO) = \sign(\OOO^*)) \to 1, \\
			&\|\widehat\LL - \LL^*\|_{\max} \lesssim_Pn^{-\frac{1}{2}+2\eta},~~~~&&\Pr(\rank(\widehat\LL) = \rank(\LL^*)) \to 1.
		\end{aligned}
	\end{equation}
\end{theorem}

Theorem \ref{thm:Omega} shows that $\widehat\OOO$ has both estimation and sign consistency, which implies that the undirected edges in ${\cal G}^*$ could be exactly reconstructed with high probability. It can also be shown that $\widehat\BB$ attains both estimation and selection consistency, implying the exact recovery of the directed edges in ${\cal G}^*$. Furthermore, given $(\widehat\OOO,\widehat\BB)$, we reconstruct the chain graph as $\widehat{\cal G}=({\cal N}, \widehat{\cal E})$, where $(i\rightarrow j) \in \widehat{\cal E}$ if and only if $\widehat{\beta}_{ji} \neq 0$, and $(i - j) \in \widehat{\cal E}$ if and only if $\widehat{\omega}_{ij} \neq 0$. The following Theorem~\ref{thm:Graph} establishes the consistency of $\widehat{\cal G}$.

\begin{theorem}\label{thm:Graph}
	Suppose all the conditions in Theorem \ref{thm:Omega} are satisfied, and we set $\kappa_n = n^{-1/2+\eta}$ and $\nu_n = n^{-1/2+2\eta}$ with a sufficiently small positive constant $\eta$. Then, we have 
 \begin{equation*}
		\begin{aligned}
			&\|\widehat\BB - \BB^*\|_{\max}  \lesssim_P n^{-\frac{1}{2}+\eta},~~~~ &&\Pr(\sign(\widehat\BB) = \sign(\BB^*)) \to 1.
		\end{aligned}
  \end{equation*}
 Furthermore, it also holds true that
 $$\Pr ( \widehat{\cal G} = {\cal G}^* ) \to 1 \ \text{as} \ n\to \infty.$$
\end{theorem}

Theorem \ref{thm:Graph} shows that $\widehat\BB$ achieves both estimation and sign consistency, which further leads to exact recovery of ${\cal G}^*$ with high probability. This is in sharp contrast to the existing methods only recovering some Markov equivalent class of the chain graph \citep{Studeny1997, Pena2014, Ma2008, Sonntag2012, Javidian2018, Pena20142, Javidian2020}.

\section{Numerical experiments}\label{sec:experiment}

\subsection{Simulated examples}

We examine the numerical performance of the proposed method \footnote{An R package ``LearnCG'' has been developed and can be downloaded from https://github.com/Ruixuan-Zhao/LearnCG.} and compare it against existing structural learning methods for chain graph, including the decomposition-based algorithm (LCD, \cite{Javidian2020}) and PC like algorithm (PC-like, \cite{Pena20142, Javidian2020}), as well as the PC algorithm for DAG (PC, \cite{Kalisch2007}). 
The implementations of LCD and PC-like are available at \url{https://github.com/majavid/AMPCGs2019}. We implement PC algorithm through R packages \texttt{pcalg} and then convert the resulted partial DAG to DAG by \texttt{pdag2dag}. The significance level of all tests in LCD, PC-like and PC is set as $\alpha = 0.05$. For the proposed method, we set $\eta=1/8$, $\gamma=2$ and $\lambda_n, \kappa_n$ and $\nu_n$ following the results in Theorems \ref{thm:Omega} and \ref{thm:Graph}.  Note that the numerical performance of the proposed method can be further refined if these tuning parameters are properly tuned via some data-adaptive scheme, at the cost of increased computational expense. 



We evaluate the numerical performances of all four methods in terms of the estimation accuracy of undirected edges, directed edges and the overall chain graph. Specifically, 
we report recall, precision and Matthews correlation coefficient (MCC) as evaluation metrics for the estimated undirected edges and directed edges, respectively. Furthermore, we employ Structural Hamming Distance (SHD) \cite{Tsamardinos2006, Javidian2020} to evaluate the estimated chain graph, which is the number of edge insertions, deletions or flips to change the estimated chain graph to the true one. Note that large values of recall, precision and MCC and small values of SHD indicate good estimation performance.


{\bf Example 1.} We consider a classic two-layer Gaussian graphical model  \cite{Lin2016, Mccarter2014} with two layers ${\cal A}_1 = \{1,..., \lceil 0.1p\rceil \}$ and ${\cal A}_2 = \{\lceil 0.1p\rceil +1, ..., p\}$, 
whose structure is illustrated in Figure \ref{SimulExample}(a). 
Within each layer, we randomly connect each pair of nodes by an undirected edge with probability $0.02$, and note that one layer may contain multiple chain components. Then, we generate the directed edges from nodes in ${\cal A}_1$ to nodes in ${\cal A}_2$ with probability $0.8$. Furthermore, the non-zero values of $\omega_{ij}$ and $\beta_{ij}$ are uniformly generated from $[-1.5,-0.5]\cup[0.5,1.5]$. To guarantee the positive definiteness of $\OOO$, each diagonal element is set as $\omega_{ii}=\sum_{j=1,j\neq i}^p \omega_{ji}+0.1$.

{\bf Example 2.} The structure of the second chain graph is illustrated in Figure \ref{SimulExample}(b). Particularly, we  randomly connect each pair of nodes by an undirected edge with probability $0.03$, and read off multiple chain components $\{\tau_1,..., \tau_m\}$ from the set of undirected edges. Then, we set the causal ordering of the chain components as $(\pi_1,...,\pi_m) = (1,...,m)$. For each chain component $\tau_k$, we randomly select the nodes as hubs with probability $0.2$, and let each hub node points to the nodes in $\cup_{i=k+1}^m \tau_i$ with probability $0.8$. Similarly, the non-zero values of $\omega_{ij}$ and $\beta_{ij}$  are uniformly generated from $[-1.5,-0.5]\cup[0.5,1.5]$, and $\omega_{ii}=\sum_{j=1,j\neq i}^p \omega_{ji}+0.1$.

\begin{figure}[!htb]
	\centering
	\begin{minipage}[b]{0.55\textwidth}
		\centering
		\includegraphics[width=1\textwidth]{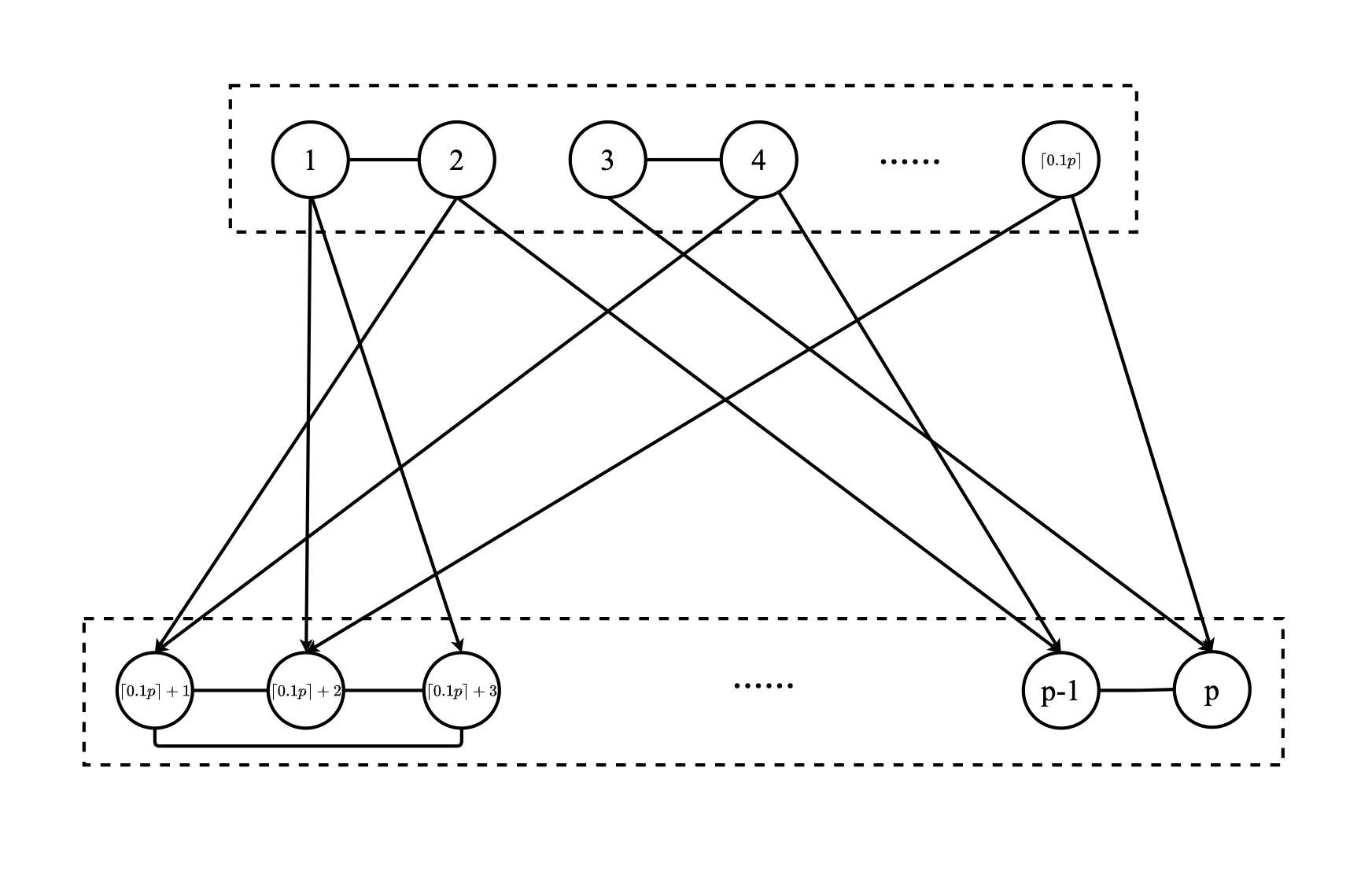}
		\subcaption{}
	\end{minipage}
	\begin{minipage}[b]{0.44\textwidth}
		\centering
		\includegraphics[width=1\textwidth]{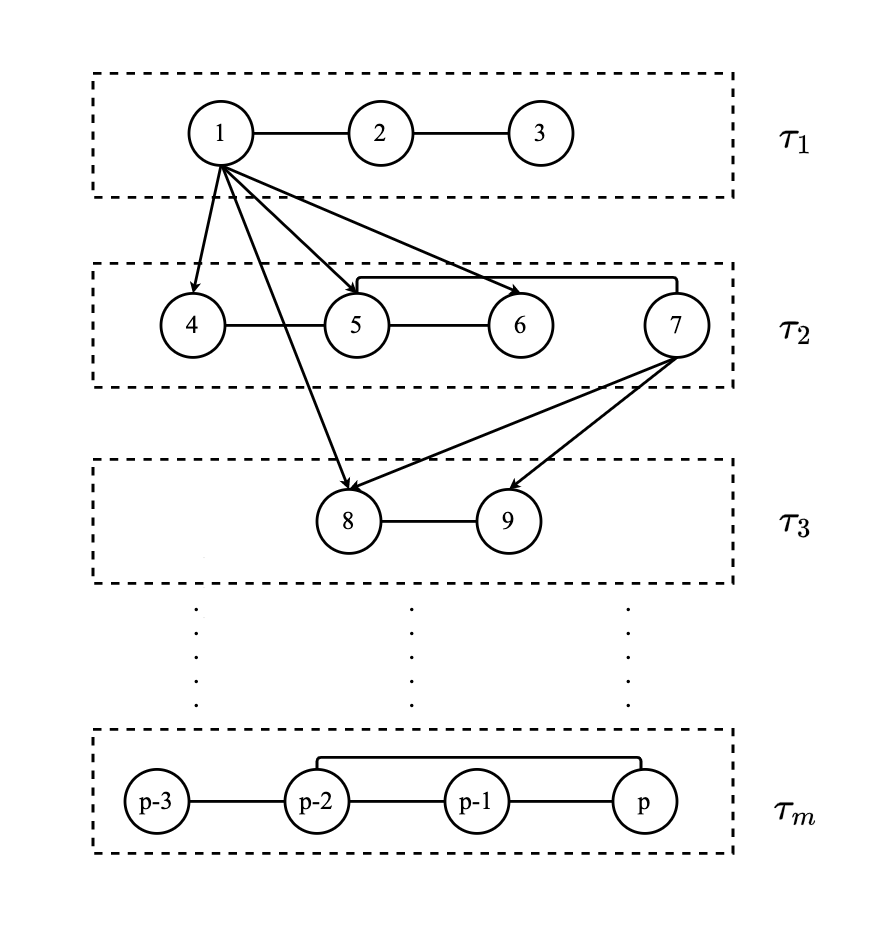}
		\subcaption{}
	\end{minipage}
	\caption{The chain graph structures in Examples 1 and 2.}\label{SimulExample}
\end{figure}

For each example, we consider four cases with $(p,n) = (50,500), (50,1000), (100,500)$ and $(100,1000)$, and the averaged performance of all four methods over 50 independent replications are summarized in Tables \ref{tab:ex1} and \ref{tab:ex2}. As PC only outputs the DAGs with no undirected edges, its evaluation metrics on $\widehat{\OOO}$ are all NA.

\begin{table}[!htb]
	\caption{The averaged evaluation metrics of all the methods in Example 1 together with
		their standard errors in parentheses.}
	\tiny
	\centering
	\begin{tabular}{ccccccccc}
		\hline 
		$ (p, n)$& Method & Recall($\widehat{\Omega}$) & Precision($\widehat{\Omega}$) & MCC($\widehat{\Omega}$) & Recall($\widehat{B}$) & Precision($\widehat{B}$) &  MCC($\widehat{B}$)   & SHD   \\\hline  
		$(50,500)$ & Proposed &\makecell[c]{\textbf{0.6478}  (0.0254)} & \makecell[c]{\textbf{0.8454}  (0.0196)} & \makecell[c]{\textbf{0.7306}  (0.0202)}  & \makecell[c]{\textbf{0.2800}  (0.0086)} & \makecell[c]{0.4963  (0.0139)} & \makecell[c]{\textbf{0.3365}  (0.0104)}  & \makecell[c]{\textbf{187.8400}  (2.5402)} \\
		& LCD  &\makecell[c]{0.1000  (0.0120)} & \makecell[c]{0.1174  (0.0128)} & \makecell[c]{0.0934  (0.0118)}  & \makecell[c]{0.1350  (0.0047)} & \makecell[c]{\textbf{0.5862}  (0.0125)} & \makecell[c]{0.2579  (0.0077)}  & \makecell[c]{192.7600  (1.5514)} \\
		&PC-like &\makecell[c]{0.0441  (0.0075)} & \makecell[c]{0.0424  (0.0084)} & \makecell[c]{0.0270  (0.0076)}  & \makecell[c]{0.0055  (0.0007)} & \makecell[c]{0.0679  (0.0090)} & \makecell[c]{-0.0006  (0.0025)}  & \makecell[c]{225.9800  (1.1569)} \\
		& PC & NA & NA & NA & \makecell[c]{0.0244  (0.0019)} & \makecell[c]{0.0835  (0.0059)} & \makecell[c]{0.0067  (0.0033)}  & \makecell[c]{238.1600  (1.1691)} \\
		\hline 
		$(50,1000)$ & Proposed &\makecell[c]{\textbf{0.6786}  (0.0230)} & \makecell[c]{\textbf{0.8309}  (0.0197)} & \makecell[c]{\textbf{0.7417}  (0.0181)}  & \makecell[c]{\textbf{0.3049}  (0.0081)} & \makecell[c]{0.4647  (0.0108)} & \makecell[c]{\textbf{0.3372}  (0.0089)}  & \makecell[c]{194.7000  (2.5856)} \\
		& LCD  &\makecell[c]{0.1174  (0.0139)} & \makecell[c]{0.1394  (0.0113)} & \makecell[c]{0.1122  (0.0115)}  & \makecell[c]{0.1869  (0.0054)} & \makecell[c]{\textbf{0.6762}  (0.0113)} & \makecell[c]{0.3324  (0.0078)}  & \makecell[c]{\textbf{181.1600}  (1.8179)} \\
		&PC-like &\makecell[c]{0.0391  (0.0067)} & \makecell[c]{0.0303  (0.0062)} & \makecell[c]{0.0161  (0.0063)}  & \makecell[c]{0.0076  (0.0009)} & \makecell[c]{0.0909  (0.0101)} & \makecell[c]{0.0054  (0.0029)}  & \makecell[c]{231.7400  (1.0909)} \\
		& PC & NA & NA & NA & \makecell[c]{0.0296  (0.0019)} & \makecell[c]{0.0898  (0.0051)} & \makecell[c]{0.0109  (0.0031)}  & \makecell[c]{242.6000  (1.3613)} \\
		\hline 
		$(100,500)$ & Proposed &\makecell[c]{\textbf{0.3491}  (0.0090)} & \makecell[c]{\textbf{0.9979}  (0.0012)} & \makecell[c]{\textbf{0.5845}  (0.0078)}  & \makecell[c]{\textbf{0.3493}  (0.0047)} & \makecell[c]{0.5391  (0.0060)} & \makecell[c]{\textbf{0.3996}  (0.0052)}  & \makecell[c]{\textbf{732.3400}  (5.2053)} \\
		& LCD  &\makecell[c]{0.0279  (0.0028)} & \makecell[c]{0.0860  (0.0089)} & \makecell[c]{0.0396  (0.0048)}  & \makecell[c]{0.0751  (0.0019)} & \makecell[c]{\textbf{0.5693}  (0.0093)} & \makecell[c]{0.1882  (0.0043)}  & \makecell[c]{794.2800  (3.0063)} \\
		&PC-like &\makecell[c]{0.0167  (0.0020)} & \makecell[c]{0.0411  (0.0058)} & \makecell[c]{0.0152  (0.0034)}  & \makecell[c]{0.0011  (0.0001)} & \makecell[c]{0.0307  (0.0039)} & \makecell[c]{-0.0084  (0.0008)}  & \makecell[c]{859.3400  (2.4988)} \\
		& PC & NA & NA & NA & \makecell[c]{0.0057  (0.0004)} & \makecell[c]{0.0421  (0.0030)} & \makecell[c]{-0.0114  (0.0011)}  & \makecell[c]{885.5600  (2.5778)} \\
		\hline
		$(100,1000)$ & Proposed &\makecell[c]{\textbf{0.4076}  (0.0102)} & \makecell[c]{\textbf{0.9982}  (0.0010)} & \makecell[c]{\textbf{0.6322}  (0.0083)}  & \makecell[c]{\textbf{0.3684}  (0.0054)} & \makecell[c]{0.4918  (0.0067)} & \makecell[c]{\textbf{0.3877}  (0.0062)}  & \makecell[c]{\textbf{770.3600}  (7.0042)} \\
		& LCD  &\makecell[c]{0.0236  (0.0022)} & \makecell[c]{0.0780  (0.0082)} & \makecell[c]{0.0340  (0.0040)}  & \makecell[c]{0.0900  (0.0022)} & \makecell[c]{\textbf{0.6349}  (0.0088)} & \makecell[c]{0.2210  (0.0045)}  & \makecell[c]{779.8200  (3.0436)} \\
		&PC-like &\makecell[c]{0.0193  (0.0020)} & \makecell[c]{0.0349  (0.0037)} & \makecell[c]{0.0131  (0.0026)}  & \makecell[c]{0.0016  (0.0002)} & \makecell[c]{0.0377  (0.0042)} & \makecell[c]{-0.0072  (0.0009)}  & \makecell[c]{872.3000  (2.4148)} \\
		& PC & NA & NA & NA & \makecell[c]{0.0077  (0.0005)} & \makecell[c]{0.0500  (0.0032)} & \makecell[c]{-0.0090  (0.0013)}  & \makecell[c]{896.0600  (2.5596)} \\
		\hline
	\end{tabular}
	\label{tab:ex1}
\end{table}

\begin{table}[!h]
	\caption{ The averaged evaluation metrics of all the  methods in Example 2 together with
		their standard errors in parentheses. Here ** denotes the fact that the corresponding methods take too long to produce any results.}
	\tiny
	\centering
	\begin{tabular}{ccccccccc}
		\hline 
		$ (p, n)$& Method & Recall($\widehat{\Omega}$) & Precision($\widehat{\Omega}$) & MCC($\widehat{\Omega}$) & Recall($\widehat{B}$) & Precision($\widehat{B}$) &  MCC($\widehat{B}$)   & SHD   \\\hline   
		$(50,500)$ & Proposed &\makecell[c]{\textbf{0.5289}  (0.0271)} & \makecell[c]{\textbf{0.7741}  (0.0161)} & \makecell[c]{\textbf{0.6244}  (0.0219)}  & \makecell[c]{\textbf{0.3230}  (0.0188)} & \makecell[c]{\textbf{0.5624}  (0.0264)} & \makecell[c]{\textbf{0.4021}  (0.0221)}  & \makecell[c]{\textbf{129.3000}  (7.3103)} \\
		& LCD  &\makecell[c]{0.4510  (0.0290)} & \makecell[c]{0.5184  (0.0261)} & \makecell[c]{0.4662  (0.0271)} & \makecell[c]{0.1646  (0.0102)} & \makecell[c]{0.5469  (0.0192)} & \makecell[c]{0.2764  (0.0108)} & \makecell[c]{132.3600  (7.9336)} \\
		&PC-like & \makecell[c]{0.4548  (0.0331)} & \makecell[c]{0.4296  (0.0257)} & \makecell[c]{0.4225  (0.0292)} & \makecell[c]{0.0231  (0.0024)} & \makecell[c]{0.1668  (0.0146)} & \makecell[c]{0.0439  (0.0052)}  & \makecell[c]{149.3600  (8.9086)} \\
		& PC & NA & NA & NA & \makecell[c]{0.1637  (0.0149)} & \makecell[c]{0.2490  (0.0140)} & \makecell[c]{0.1655  (0.0130)}  & \makecell[c]{160.1200  (8.2415)} \\
		\hline 
		$(50,1000)$ & Proposed & \makecell[c]{\textbf{0.5790}  (0.0301)} & \makecell[c]{\textbf{0.7467}  (0.0172)} & \makecell[c]{\textbf{0.6421}  (0.0237)}  & \makecell[c]{\textbf{0.3372}  (0.0215)} & \makecell[c]{0.5446  (0.0286)} & \makecell[c]{\textbf{0.4021}  (0.0246)}  & \makecell[c]{130.4200  (8.1150)} \\
		& LCD  &\makecell[c]{0.4723  (0.0301)} & \makecell[c]{0.4873  (0.0218)} &  \makecell[c]{0.4609  (0.0255)} & \makecell[c]{0.1885  (0.0110)} & \makecell[c]{\textbf{0.5721}  (0.0173)} &  \makecell[c]{0.3052  (0.0119)} & \makecell[c]{\textbf{128.7800}  (7.8959)} \\
		&PC-like &\makecell[c]{0.4584  (0.0309)} & \makecell[c]{0.4108  (0.0239)} & \makecell[c]{0.4138  (0.0269)} & \makecell[c]{ 0.0205  (0.0025)} & \makecell[c]{ 0.1529  (0.0181)} & \makecell[c]{0.0379  (0.0059)} & \makecell[c]{149.9800  (8.8576)} \\
		& PC & NA & NA & NA & \makecell[c]{0.1848  (0.0176)} & \makecell[c]{0.2599  (0.0142)} & \makecell[c]{0.1828  (0.0156)} & \makecell[c]{159.1600  (8.5205)} \\
		\hline 
		$(100,500)$ & Proposed &\makecell[c]{\textbf{0.3723}  (0.0094)} & \makecell[c]{ \textbf{0.9985}  (0.0009)} & \makecell[c]{\textbf{0.6015}  (0.0077)} & \makecell[c]{\textbf{0.4211}  (0.0260)} & \makecell[c]{\textbf{0.7569}  (0.0256)} & \makecell[c]{\textbf{0.5536}  (0.0253)} & \makecell[c]{\textbf{151.3800}  (4.3586)} \\
		& LCD  & ** & ** & ** & ** & ** & ** &  **  \\
		&PC-like & ** & ** & ** & ** & ** & ** &  **  \\
		& PC & NA & NA & NA & \makecell[c]{0.2710  (0.0269)} & \makecell[c]{0.1013  (0.0087)} & \makecell[c]{0.1472  (0.0125)} & \makecell[c]{231.6000  (5.1892)} \\
		\hline
		$(100,1000)$ & Proposed &\makecell[c]{\textbf{0.5515}  (0.0100)} & \makecell[c]{\textbf{0.9971}  (0.0012)} & \makecell[c]{\textbf{0.7349}  (0.0069)} & \makecell[c]{\textbf{0.6230}  (0.0254)} & \makecell[c]{\textbf{0.8984}  (0.0215)} & \makecell[c]{\textbf{0.7410}  (0.0227)} & \makecell[c]{\textbf{103.7400}  (4.0098)} \\
		& LCD  & ** & ** & ** & ** & ** & ** &  ** \\
		&PC-like & ** & ** & ** & ** & ** & ** &  **  \\
		& PC & NA & NA & NA & \makecell[c]{0.2875  (0.0264)} & \makecell[c]{0.1065  (0.0091)} & \makecell[c]{0.1562  (0.0127)} & \makecell[c]{228.1400  (4.9927)} \\
		\hline
	\end{tabular}
	\label{tab:ex2}
\end{table}

From Tables \ref{tab:ex1} and \ref{tab:ex2}, it is clear that the proposed method outperforms all competitors in most scenarios.  
In Example 1, the proposed method produces a much better estimation of the undirected edges than all other methods. For directed edges, the proposed method achieves the highest Recall$(\widehat{\BB})$ and MCC$(\widehat{\BB})$. It is interesting to note that LCD gets higher Precision$(\widehat{\BB})$ than the proposed method, possibly due to the fact that LCD tends to produce fewer estimated directed edges, resulting in large Precision$(\widehat{\BB})$ but small Recall$(\widehat{\BB})$. In Example 2, the proposed method outperforms all competitors in terms of almost all the evaluation metrics. Note that LCD and PC-like take too long to produce any results when $p=100$, due to their expensive computational cost when there exist many hub nodes.


\subsection{Standard \& Poor index data}

We apply the proposed method to study the relationships among stocks in the Standard \& Poor's 500 index, and analyze the impact of the COVID-19 pandemic on the stock market. Chain graph can accurately reveal various relationships among stocks, with undirected edges for symmetric competitive or cooperative relationship between stocks and directed edges for asymmetric causal relation from one stock to the another. 

To proceed, we select $p = 100$ stocks with the largest market sizes in the Standard \& Poor's 500 index, and retrieve their adjusted closing prices during the pre-pandemic period, August 2017-February 2020, and the post-pandemic period, March 2020-September 2022.  The data is publicly available on many finance websites and has been packaged in some standard softwares, such as the R package \texttt{quantmod}. For each period, we first calculate the daily returns of each stock based on its adjusted closing prices, and then apply the proposed method to construct the corresponding chain graph. 

Figure \ref{fig:financeUG} displays the undirected edges between stocks in both estimated chain graphs, which consist of $39$ and $21$ undirected edges in pre-pandemic and post-pandemic, respectively. It is clear that 
there are more estimated undirected edges in the 
pre-pandemic chain graph than in the post-pandemic one, which echos the empirical findings that business expansion were more active, company cooperations were closer and competition were fiercer before the COVID-19 pandemic. Furthermore, there are $13$ common undirected edges in both chain graphs, and all these 13 connected stock pairs are from the same sector, including VISA(V) and MASTERCARD(MA), JPMORGAN CHASE(JPM) and BANK OF AMERICA(BAC), MORGAN STANLEY(MS) and GOLDMAN SACHS(GS), and HOME DEPOT(HD) and LOWE'S(LOW).  All these pairs share the same type of business, and their competition or cooperation receive less impact by the COVID-19 pandemic. In Figure \ref{fig:financeUG}, it is also interesting to note that the number of the undirected edges between stocks from different sectors have reduced in the post-pandemic chain graph. This concurs with the fact that diversified business transactions between companies have been decreased and only essential business contacts have been maintained during the COVID-19 pandemic.

\begin{figure}
	\centering
	\includegraphics[width=1\textwidth]{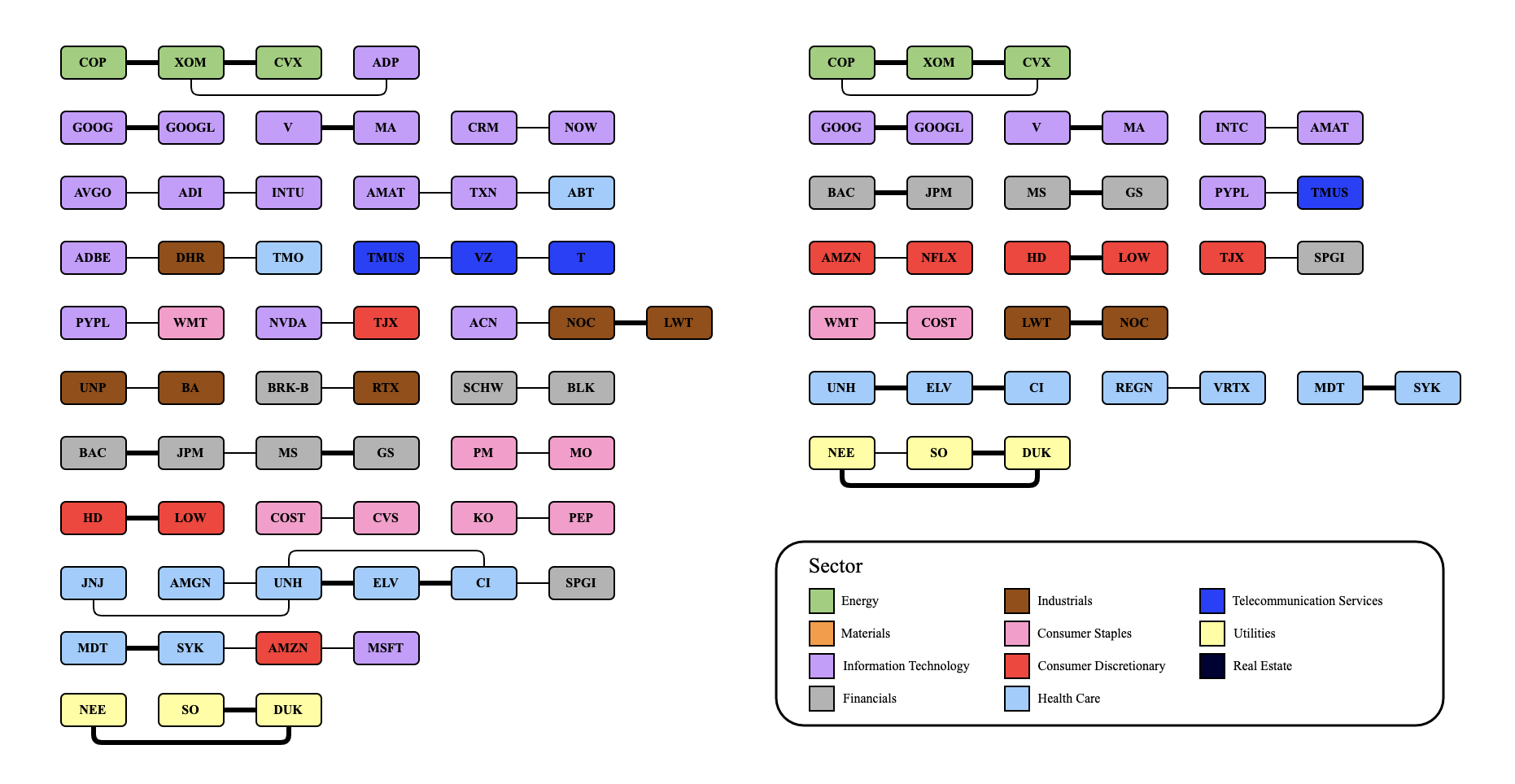}
	\caption{The left and right panel display all the estimated undirected edges for pre-pandemic and post-pandemic, respectively. Stocks from the same sectors are dyed with the same color, and the common undirected edges in both chain graphs are boldfaced.}\label{fig:financeUG}
\end{figure}

Figure \ref{fig:financeOrder} displays the boxplots of causal orderings of all stocks within each sector in both pre-pandemic and post-pandemic, where the causal ordering of a stock is set as that of the corresponding chain component. It is generally believed that causal ordering implies the imbalance of social demand and supply; that is, if a sector is getting more demanded, its causal ordering is inclined to move up to upstream. Evidently, Energy and Materials are always at the top of the causal ordering in both periods, as they are upstream industries and provide inputs for most other sectors. The median causal ordering of Telecommunication Services goes from downstream to upstream after the outbreak of the COVID-19 pandemic, since people travel less and rely more on telecommunication for business communication. The median causal ordering of Finances goes down during the pandemic, as commercial entities are more cautious about credit expansion and demand for financial services is likely to decline to battle the financial uncertainty. It is somewhat surprising that the median causal ordering of Healthcares appears invariant, but many pharmaceutical and biotechnology corporation in this section actually have changed from downstream to upstream, due to the rapid development vaccine and treatment during the pandemic. 

\begin{figure}[!htb]
	\centering
	\begin{minipage}[b]{0.49\textwidth}
		\centering
		\includegraphics[width=1\textwidth]{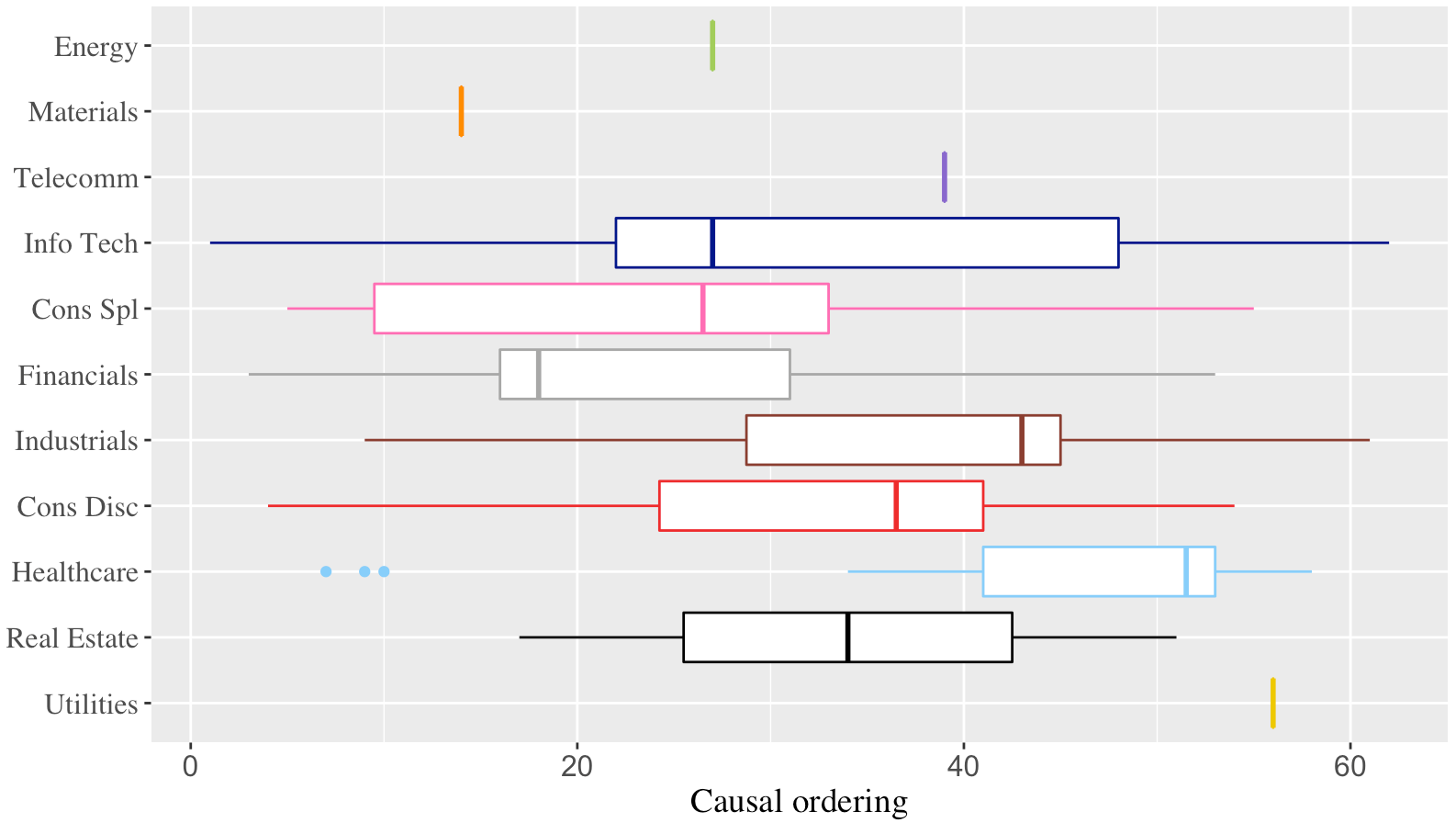}
	\end{minipage}
	\begin{minipage}[b]{0.49\textwidth}
		\centering
		\includegraphics[width=1\textwidth]{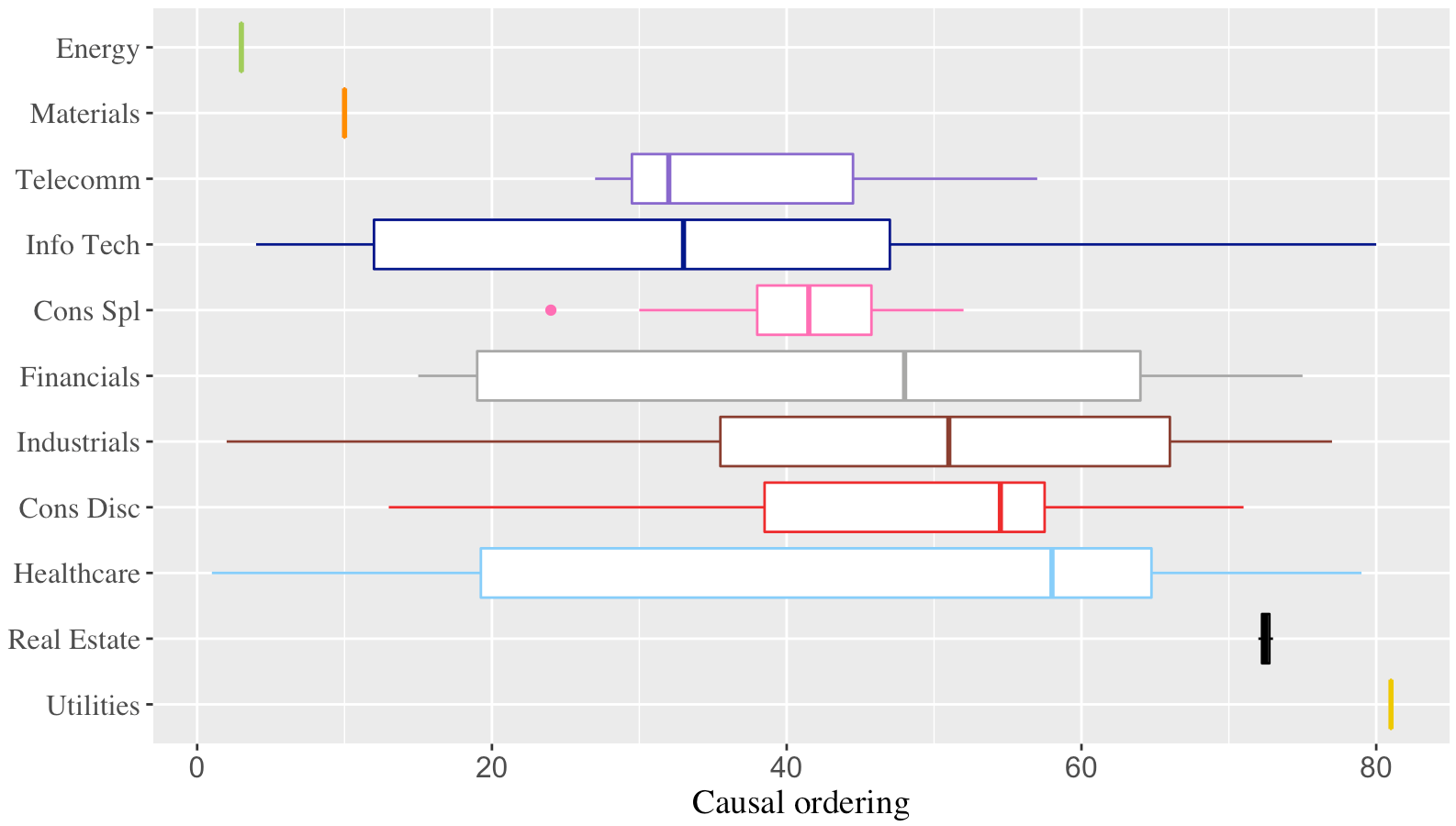}
	\end{minipage}
	\caption{The left and right panel display the boxplots of the estimated causal ordering of the top 100 stocks in each sector for pre-pandemic and post-pandemic, respectively. The sectors are ordered according to the median causal ordering of stocks in post-pandemic.}\label{fig:financeOrder}
\end{figure}

In addition, the estimated chain graphs in pre-pandemic and post-pandemic consist of $149$ and $190$ directed edges, respectively. While many directed edges remain unchanged, there are some stocks whose roles have changed dramatically in the chain graphs. In particular, some stocks with no child but multiple parents in the pre-pandemic chain graph become ones with no parent but multiple children in the post-pandemic chain graph, such as COSTCO(COST), APPLE(AAPL), ACCENTURE(ACN), INTUIT(INTU), AT\&T(T) and CHUBB(CB). This finding appears reasonable, as most of these stocks correspond to the high demanded industries during pandemic, such as COSTCO for stocking up groceries, AT\&T for remote communication, and APPLE for providing communication and online learning equipment. On the other hand, there are some other stocks with no parent but multiple children in the pre-pandemic chain component  becoming
ones with no child but multiple parents in the post-pandemic chain component, including TESLA(TSLA), TJX(TJX), BRISTOL-MYERS SQUIBB(BMY), PAYPAL(PYPL), AUTOMATIC DATA PROCESSING(ADP) and BOEING(BA). Many of these companies have been severely impacted during pandemic, such as BOEING due to minimized travels and TESLA due to shrunk consumer purchasing power.

\section{Conclusion}\label{sec:conclusion}

In this paper, we establish a set of novel identifiability conditions for the Gaussian chain graph model under AMP interpretation, exploiting a low rank plus sparse decomposition of the precision matrix.
An efficient learning algorithm is developed to recover the exact chain graph structure, including both undirected and directed edges. Theoretical analysis shows that the proposed method consistently reconstructs
the exact chain graph structure. Its advantage is also supported by various numerical experiments on both simulated and real examples. It is also interesting to extend the proposed identifiability conditions and learning algorithm to accommodate non-linear chain graph model with non-Gaussian noise.

\section*{Acknowledgment}
This work is supported in part by HK RGC Grants GRF-11304520, GRF-11301521 and GRF-11311022. The authors report there are no competing interests to declare.

\section*{Appendix: Proof of theorems}

In the sequel, we use $c$ and $C$ to denote generic positive constants whose values may vary according to context.  For a matrix $\MM\in\R^{p\times p},$ let $\OO(\MM)\in\R^{p\times p}$ denote the matrix with the same off-diagonal elements as $\MM$ but all diagonal elements being 0. 
Denote $\Lambda_{\max}(\MM)$ and $\Lambda_{\min}(\MM)$ as the maximal and minimal eigenvalue of a matrix $\MM$, respectively.  Let $\AAA$ be the matrix addition operator such that $\AAA(\A,\BB) = \A+\BB$ for two matrices $\A$ and $\BB$ with the same dimension.

By the definition of ${\cal T}(\LL)$ in \eqref{eq:tangent L}, it can be shown that ${\cal T}(\LL)$ is uniquely determined by $\UU_1\in\R^{p\times K},$ where $\LL = \UU_1 \DD_1 \UU_1^\top$ is the eigen decomposition of $\LL.$ With a slight abuse of notation, we denote ${\cal T}(\UU_1) = {\cal T}(\LL).$  Further, define $F_{\UU_1}:{{\cal S}(\OOO^*)}\times {\cal T}(\UU_1)\to {{\cal S}(\OOO^*)}\times {\cal T}(\UU_1)$ as
$$
F_{\UU_1}(\DDD_{\OOO},\DDD_{\LL}) = (\PPP_{{{\cal S}(\OOO^*)}}\III^*v(\DDD_{\OOO}+\DDD_{\LL}), \PPP_{{\cal T}(\UU_1)}\III^*v(\DDD_{\OOO}+\DDD_{\LL})),
$$ 
and $F_{\UU_1}^\perp:{{\cal S}(\OOO^*)}\times {\cal T}(\UU_1)\to {{\cal S}(\OOO^*)}^\perp\times {\cal T}(\UU_1)^\perp$ as
$$
F_{\UU_1}^\perp(\DDD_{\OOO},\DDD_{\LL}) = (\PPP_{{{\cal S}(\OOO^*)}^\perp}\III^*v(\DDD_{\OOO}+\DDD_{\LL}), \PPP_{{\cal T}(\UU_1)^\perp}\III^*v(\DDD_{\OOO}+\DDD_{\LL})).
$$ 
For any $\UU_1\in\R^{p\times K},$ define $D(\UU_1)= \{\UU_{1}\DD_1\UU_{1}^\top:\DD_1 \text{ is a } K\times K \text{ diagonal matrix}\}.$ Recall the operators $F$ and $F^{\perp}$ defined in Section~\ref{sec:theory}, and we have $F_{\UU_1^*} = F$ and $F_{\UU_1^*}^{\perp} = F^{\perp}.$

Let $\OS_0 = \{(\OOO,\LL)\in\R^{p\times 2p}:\OOO \succ 0,~~\OOO+\LL \succ0\}$, and define a localization set,
$$
\begin{aligned}
	\OS(\epsilon) = \{(\OOO,\LL) & \subset\OS_0: \OOO = \OOO^* +\DDD_{\OOO}, \DDD_{\OOO}\in\SSSS(\OOO^*), \|\DDD_{\OOO}\|_{\max}\leq\epsilon,\\
	&\LL = \UU_1\DD_1\UU_1^\top, \UU_1 \text{ is } p\times K \text{ orthogonal matrix with } \|\UU_1 - \UU_1^*\|_{\max} \leq \epsilon, \\
	&\DD_1 \text{ is } K\times K \text{ diagonal matrix with } \|\DD_1 - \DD_1^*\|_{\max} \leq \epsilon \}.
\end{aligned}
$$ 
We rewrite $l(\OOO+\LL)$ as $l(\OOO,\LL),$ and define $\bar l(\OOO,\LL) = \EEE l(\OOO+\LL)$. Proposition \ref{prop:ident} is a key intermediate result to the proof of Theorem \ref{thm:ident}.

\begin{proposition}\label{prop:ident}
	Suppose $(\OOO^*,\BB^*)\in\PA.$ Then, there exists a small constant $\epsilon>0$ such that $(\OOO^*,\LL^*)$ is the unique maximizer of $\bar l(\OOO,\LL)$ in $\OS(\epsilon).$
\end{proposition}

\begin{proof}[\bf Proof of Proposition~\ref{prop:ident}]
The proof is deferred to the supplement.
\end{proof}

\begin{proof}[\bf Proof of Theorem \ref{thm:ident}]

The proof proceeds in two steps, and we first show that $\OOO = \OOO^*$. By Proposition~\ref{prop:ident}, there exists an $\epsilon_0$ such that $(\OOO^*,\LL^*)$ is the unique maximizer of $\bar l(\OOO,\LL)$ in $\OS(\epsilon_0).$

Note that for any $\epsilon>0$ and
$(\OOO,\BB)\in\PA$ satisfying $\|\OOO-\OOO^*\|_{\max}<\epsilon$ and $\|\BB-\BB^*\|_{\max}<\epsilon,$ it implies that $\|\LL-\LL^*\|_{\max}\lesssim\epsilon.$ Then, by Weyl's inequality, we have $\|\DD_1-\DD_1^*\|_{\max}\lesssim\epsilon.$ 
By the Davis-Kahan theorem \cite[][]{yu2015useful} and Assumption~\ref{ass:diff}, we have $\|\UU_1-\UU_1^*\|_{\max}\lesssim\epsilon.$


Therefore, there exists a sufficiently small $\epsilon>0$ such that for any $(\OOO,\BB)\in\PA$ satisfying $\|\OOO-\OOO^*\|_{\max}<\epsilon$ and $\|\BB-\BB^*\|_{\max}<\epsilon,$ we have $(\OOO,\LL)\in\OS(\epsilon_0).$
If $(\mathbf{I}_p-\BB)^\top\OOO (\mathbf{I}_p-\BB)=(\mathbf{I}_p-\BB^*)^\top \OOO^* (\mathbf{I}_p-\BB^*)$, 
we have $\OOO+\LL = \OOO^*+\LL^*,$ implying that $\bar l(\OOO,\LL) = \bar l(\OOO^*,\LL^*).$ By the fact that $(\OOO^*,\LL^*)$ is the unique maximizer of $\bar l(\OOO,\LL)$ in $\OS(\epsilon_0),$ it must hold true that $(\OOO,\LL) = (\OOO^*,\LL^*)$.

Next we turn to show that $\BB = \BB^*$. Note that $\OOO=\OOO^*$ leads to the same set of chain components $\{\tau_k\}_{k=1}^m$, and the assumption that $(\mathbf{I}_p-\BB)^\top\OOO (\mathbf{I}_p-\BB)=(\mathbf{I}_p-\BB^*)^\top \OOO^* (\mathbf{I}_p-\BB^*)$ leads to $\mathbf{\Sigma}^*={\bf\Sigma}$. Furthermore, let
\begin{align}\label{criter-pop}
	{\cal D} (\tau_k,{\cal C}) :=& \max_{i\in \tau_k} \Big\{ \Var (x_{i} | \xx_{\cal C}) - \OOO_{ii}^{-1} \Big\}\nonumber\\
	=& \max_{i\in \tau_k} \Big\{ \mathbf{\Sigma}_{ii} - \mathbf{\Sigma}_{i{\cal C}} \mathbf{\Sigma}_{{\cal C}{\cal C}}^{-1} \mathbf{\Sigma}_{{\cal C}i} -  \OOO_{ii}^{-1} \Big\},	
\end{align}
and ${\cal D} (\tau_k,\emptyset) := \max_{i\in \tau_k} \Big\{ \mathbf{\Sigma}_{ii} - \OOO_{ii}^{-1} \Big\}$. The Lemma 3 in the supplement implies that ${\cal D}(\tau_{\pi_1},\emptyset)=0$ for a root chain component $\tau_{\pi_1}$ with $\pa(\tau_{\pi_1})=\emptyset$, and that ${\cal D}(\tau_{k},\emptyset)>0$ for any other chain component $\tau_k$ with $\pa(\tau_{k})\neq\emptyset$. Further, let ${\cal C}_j = \cup_{k=1}^j \tau_{\pi_k}$ for any $j \ge 1$, and the Lemma 3 in the supplement implies that ${\cal D}(\tau_{\pi_{j+1}},{\cal C}_{j})=0$ for the $(j+1)$-th chain component $\tau_{\pi_{j+1}}$, and ${\cal D}(\tau_{k},{\cal C}_{j})>0$ for any chain component $\tau_k$ with $\pa(\tau_{k})\nsubseteq {\cal C}_j$. Therefore, the causal ordering among $\{\tau_k\}_{k=1}^m$, denoted as $\{\pi_k\}_{k=1}^m$, can be recovered via mathematical induction.

Corresponding to $\{\tau_{\pi_1},...,\tau_{\pi_m}\}$, there exists a permutation matrix $\PP_\pi$ such that 
$$
\PP_\pi \OOO^* \PP_\pi^\top = \PP_\pi \OOO \PP_\pi^\top = \diag \big (\OOO_{\tau_{\pi_1},\tau_{\pi_1}},...,\OOO_{\tau_{\pi_m},\tau_{\pi_m}} \big ).
$$
Since $(\OOO,\BB), (\OOO^*,\BB^*)\in {\cal Q}$, $\PP_\pi \BB \PP_\pi^\top$ and $\PP_\pi \BB^* \PP_\pi^\top$ are block lower triangular matrices with $m$ zero diagonal blocks, which implies that the submatrices $\BB_{\tau_{\pi_k},\tau_{\pi_j}}$ and $\BB^*_{\tau_{\pi_k},\tau_{\pi_j}}$ are zeros for $k\geq j$.  Then,  for $(\OOO,\BB)\in {\cal Q}$, it follows from (\ref{eq:model}) that 
$$\xx_{\tau_{\pi_k}} = \BB_{\tau_{\pi_k}, {\cal C}_{k-1}} \xx_{{\cal C}_{k-1}} + \mathbf{\epsilon}_{\tau_{\pi_k}}$$ with $\xx_{{\cal C}_{k-1}} \indep \mathbf{\epsilon}_{\tau_k}$, which implies $ \BB_{\tau_{\pi_k}, {\cal C}_{k-1}} = \mathbf{\Sigma}_{\tau_{\pi_k}, {\cal C}_{k-1}} \mathbf{\Sigma}_{{\cal C}_{k-1}, {\cal C}_{k-1}}^{-1}$. Similarly, for $(\OOO^*,\BB^*) \in {\cal Q}$, we obtain $ \BB^*_{\tau_{\pi_k}, {\cal C}_{k-1}} =  \mathbf{\Sigma}^*_{\tau_{\pi_k}, {\cal C}_{k-1}} (\mathbf{\Sigma}^*_{{\cal C}_{k-1}, {\cal C}_{k-1}})^{-1}$. Since $\mathbf{\Sigma} = \mathbf{\Sigma}^*$, $\BB_{\tau_{\pi_k}, {\cal C}_{k-1}}  = \BB^*_{\tau_{\pi_k}, {\cal C}_{k-1}} $ for any $2\leq k\leq m$. Therefore, $\BB=\BB^*$, which completes the proof of Theorem \ref{thm:ident}.
\end{proof}

\begin{proof}[\bf Proof of Theorem \ref{thm:Omega}] Let $h(\OOO,\LL) = -l(\OOO+\LL) + \lambda_n\big(\|\OOO\|_{1,\off} + \gamma\|\LL\|_* \big),$ 
	$$
	\begin{aligned}
		\OS_1 = \{(\OOO,\LL)\subset\OS_0: ~&\OOO = \OOO^* +\DDD_{\OOO}, \DDD_{\OOO} \text{ is } p\times p \text{ symmetric matrix  with } \|\DDD_{\OOO}\|_{\max} \leq \lambda_n^{1-2\eta},\\
		&\LL = \UU\DD\UU^\top,  \UU \text{ is } p\times p \text{ orthogonal matrix with } \|\UU - \UU^*\|_{\max} \leq \lambda_n^{1-\eta}, \\
		& \DD \text{ is } p\times p \text{ diagonal matrix with } \|\DD-\DD^*\|_{\max} \leq \lambda_n^{1-2\eta} \}
	\end{aligned}
	$$ 
	with $\eta>0$, and $\OS_2 = \{ (\OOO,\LL)\in\OS_1: \OOO-\OOO^* \in {{\cal S}(\OOO^*)}, \rank(\LL)\leq K \}$. 
	
	By the compactness of $\OS_2$,  $h(\OOO,\LL)$ has a minimizer in $\OS_2$, denoted as $(\widehat\OOO_{\OS_2},\widehat\LL_{\OS_2})$, and let $\widehat\LL_{\OS_2} = \widehat\UU_{\OS_2}\widehat\DD_{\OS_2}\widehat\UU_{\OS_2}^\top$ be its eigen decomposition, where $\widehat\DD_{\OS_2}$ is a $K\times K$ diagonal matrix. Further, let $\widehat\OS = {{\cal S}(\OOO^*)} \times D(\widehat\UU_{\OS_2}).$ By Lemma 4, with probability approaching 1, $h(\OOO,\LL)$ has a unique minimizer in $\widehat\OS,$ denoted as $(\widehat\OOO_{\widehat\OS},\widehat\LL_{\widehat\OS})$, and we have $(\widehat\OOO_{\widehat\OS},\widehat\LL_{\widehat\OS}) \in \OS_2$. By the definition of $\widehat\OS,$ we also have $(\widehat\OOO_{\OS_2},\widehat\LL_{\OS_2}) \in \widehat\OS.$ Then, $h(\widehat\OOO_{\widehat\OS},\widehat\LL_{\widehat\OS}) = h(\widehat\OOO_{\OS_2},\widehat\LL_{\OS_2})$, which immediately implies that $(\widehat\OOO_{\OS_2},\widehat\LL_{\OS_2}) = (\widehat\OOO_{\widehat\OS},\widehat\LL_{\widehat\OS})$, 
	\begin{align}
		\label{eq:SD inactive} 
		\|\widehat\OOO_{\OS_2} - \OOO^*\|_{\max} <_P \lambda_n^{1-\eta}, \mbox{ and }
		\|\widehat\DD_{\OS_2} - \DD_1^*\|_{\max} <_P \lambda_n^{1-\eta}. 
	\end{align}
	
	We now turn to bound $\|\widehat\UU_{\OS_2} - \UU_1^*\|_{\max}$. For any $\DDD\in\R^{p\times p}$, let
	\begin{equation}
		\label{eq:taylor1}
		R_n(\DDD) = l(\TTT^*+\DDD) - \bar l(\TTT^*) + \frac{1}{2}v(\DDD)^\top\III^*v(\DDD).
	\end{equation}
	Then, as $n\to\infty$ and $\|\DDD\|_{\max}\to0,$ we have 
	\begin{equation}\label{eq:remainder1}
		\begin{aligned}
			|R_n(\DDD) - R_n(\textbf{0})| &= |l(\TTT^*+\DDD) - l(\TTT^*)+\frac{1}{2}v(\DDD)^\top\III^*v(\DDD)| \\
			&=  |\langle\nabla l(\TTT^*),\DDD\rangle| + O(\|\DDD\|^3_{\max})\\
			&= O_p(\|\DDD\|_{\max}/\sqrt{n}) + O(\|\DDD\|_{\max}^3),
		\end{aligned}
	\end{equation}
	where the last equality follows from the fact that $\|\nabla l(\TTT^*)\|_{\max} = \|\widehat\SG - (\TTT^*)^{-1}\|_{\max} = O_p(1/\sqrt{n})$ by the law of large number. Similarly, it can be verified that
	\begin{align}
		|\nabla R_n(\DDD)| &= O_p(1/\sqrt{n})+O(\|\DDD\|_{\max}^2), \label{eq:remainder2} \\
		|\nabla^2 R_n(\DDD)| &= O(\|\DDD\|_{\max}). \label{eq:remainder3}
	\end{align}
	Particularly, let $\DDD = \widehat\OOO_{\widehat\OS}+\widehat\LL_{\widehat\OS} - \OOO^* - \LL^*$, 
	\begin{equation*}
		\begin{aligned}
			\|\DDD\|_{\max} &= \|\widehat\OOO_{\widehat\OS}+\widehat\LL_{\widehat\OS} - \OOO^* - \LL^*\|_{\max} \geq \|\widehat\OOO_{\widehat\OS}- \OOO^*+\widehat\DDD_{\LL}\|_{\max} - \|\widehat\LL_{\widehat\OS} - \LL^*-\widehat\DDD_{\LL}\|_{\max}\\
			&\gtrsim \|\widehat\DDD_{\LL}\|_{\max} - \|\widehat\LL_{\widehat\OS} - \LL^*-\widehat\DDD_{\LL}\|_{\max} >_{P} c\|\widehat\UU_{\OS_2}-\UU_1^*\|_{\max} - c\lambda_n^{2(1-\eta)}.
		\end{aligned}
	\end{equation*}
	where $\widehat\DDD_{\LL} = \UU_1^*\DD_1^*(\widehat\UU_{\OS_2} - \UU_1^*)^\top+(\widehat\UU_{\OS_2}-\UU_1^*)\DD_1^*(\UU_1^*)^\top + \UU_1^*(\widehat\DD_{\widehat\OS} -\DD_1^*)(\UU_1^*)^\top$, the second inequality is due to Lemma 2, and the third inequality is due to Lemma 6. By the fact that $\III^*$ is positive definite, we have 
	$$
	\frac{1}{2}v(\DDD)^\top\III^*v(\DDD)\geq \frac{1}{2}\Lambda_{min}(\III^*) \|\DDD\|_{\max}^2 >_p c\|\widehat\UU_{\OS_2}-\UU_1^*\|_{\max}^2 - c\lambda_n^{4(1-\eta)}.
	$$ 
	On the other hand, we also have
	$$
	\begin{aligned}
		& \frac{1}{2}v(\DDD)^\top\III^*v(\DDD) = [R_n(\DDD)- R_n({\bf0})] - [l(\widehat\OOO_{\widehat\OS}+\widehat\LL_{\widehat\OS}) - l(\OOO^* +\LL^*)]\\
		&\leq [R_n(\DDD)- R_n({\bf0})] + [h(\widehat\OOO_{\widehat\OS},\widehat\LL_{\widehat\OS}) - h(\OOO^*,\LL^*)] + \lambda_n\left\{\left| \|\widehat\OOO_{\widehat\OS}\|_{1,\off} - \|\OOO^*\|_{1,\off}\right| + \left|\gamma\|\widehat\LL_{\widehat\OS}\|_* -  \gamma \|\LL^*\|_*\right|\right\} \\
		&\leq [R_n(\DDD)- R_n({\bf0})] + \lambda_n\left\{\left| \|\widehat\OOO_{\widehat\OS}\|_{1,\off} - \|\OOO^*\|_{1,\off}\right| + \left|\gamma\|\widehat\LL_{\widehat\OS}\|_* -  \gamma \|\LL^*\|_*\right|\right\} \\
		&\lesssim_P [R_n(\DDD)- R_n({\bf0})] + \lambda_n^{2-\eta} = O_p(\|\DDD\|_{\max}/\sqrt{n}) +O(\|\DDD\|_{\max}^3 + \lambda_n^{2-\eta}),
	\end{aligned}
	$$ where the last inequality is due to \eqref{eq:remainder1}.
 	Therefore, together with the fact that $\|\DDD\|_{\max}\lesssim_P\lambda_n^{1-\eta}$, we obtain 
	$$
	\|\widehat\UU_{\OS_2}-\UU_1^*\|_{\max}=O_p(n^{-\frac{1}{4}}\lambda_n^{\frac{1}{2}(1-\eta)}+\lambda_n^{1-\frac{\eta}{2}}) <_P\lambda_n^{1-\eta}.
	$$ 
	
	Combing the upper bounds for $\|\widehat\OOO_{\OS_2} - \OOO^*\|_{\max}$, 
	$\|\widehat\DD_{\OS_2} - \DD_1^*\|_{\max}$ and $\|\widehat\UU_{\OS_2}-\UU_1^*\|_{\max}$, we conclude that $(\widehat\OOO_{\OS_2},\widehat\LL_{\OS_2})$ lies in the interior of $\OS_2.$ Further, let $\DDD_{\OS_2} = \widehat\OOO_{\OS_2}+\widehat\LL_{\OS_2} - \OOO^* - \LL^*$, and $\VV^*_1\in\R^{K\times p}$ contain the right singular vectors of $\LL^*$, then $\VV_1^* =\UU_1^* \sign(\DD_1^*).$ Similar as Lemma 5 in \cite{chen2016fused}, we have that $(\widehat\OOO_{\OS_2},\widehat\LL_{\OS_2})$ is the unique minimizer of $h(\OOO,\LL)$ in $\OS_2$, and that
	\begin{equation}\label{eq:res decomp}
		\DDD_{\OS_2} = \lambda_n\AAA\left[ F_{\UU_1^*}^{-1}(\sign(\OO(\OOO^*)),\gamma\UU_1^*(\VV_1^*)^\top)\right] + o_p(\lambda_n).
	\end{equation}
	
	Next, we show that $(\widehat\OOO_{\OS_2},\widehat\LL_{\OS_2})$ is also a minimizer of $h(\OOO,\LL)$ in $\OS_1$ by verifying the first order condition. Since $(\widehat\OOO_{\OS_2},\widehat\LL_{\OS_2})$ is the minimizer of $h(\OOO,\LL)$ in $\OS_2$ and lies in its interior,  
	$$
	\0 \in \PPP_{{{\cal S}(\OOO^*)}}\partial_{\OOO} h(\widehat\OOO_{\OS_2},\widehat\LL_{\OS_2}), \text{~~and~~} \0 \in \PPP_{{\cal T}(\widehat\UU_{\OS_2})}\partial_{\LL} h(\widehat\OOO_{\OS_2},\widehat\LL_{\OS_2}).
	$$ 
	Then, it suffices to show that 
	\begin{equation}
		\label{eq:orthogonal cond}
		\0 \in \PPP_{{{\cal S}(\OOO^*)}^\perp}\partial_{\OOO} h(\widehat\OOO_{\OS_2},\widehat\LL_{\OS_2}), \text{~~and~~} \0 \in \PPP_{{\cal T}(\widehat\UU_{\OS_2})^\perp}\partial_{\LL} h(\widehat\OOO_{\OS_2},\widehat\LL_{\OS_2}).
	\end{equation}
	
	Taking derivatives on both sides of \eqref{eq:taylor1} and substituting $\DDD$ by $\DDD_{\OS_2}$, it follows from \eqref{eq:remainder2} and \eqref{eq:res decomp} that
	$$
	\nabla l(\widehat\OOO_{\OS_2}+\widehat\LL_{\OS_2}) = \lambda_n\III^*\left(\AAA\left[ F_{\UU_1^*}^{-1}(\sign(\OO(\OOO^*)),\gamma\UU_1^*(\VV_1^*)^\top)\right]\right) + o_p(\lambda_n),
	$$ 
	which further implies that 
	\begin{align*}
		\PPP_{{{\cal S}(\OOO^*)}^\perp}\nabla l(\widehat\OOO_{\OS_2}+\widehat\LL_{\OS_2}) &= \lambda_n\PPP_{{{\cal S}(\OOO^*)}^\perp}\III^*\left(\AAA\left[ F_{\UU_1^*}^{-1}(\sign(\OO(\OOO^*)),\gamma\UU_1^*(\VV_1^*)^\top)\right]\right) + o_p(\lambda_n), \\
		\PPP_{{\cal T}(\widehat\UU_{\OS_2})^\perp}\nabla l(\widehat\OOO_{\OS_2}+\widehat\LL_{\OS_2}) &= \lambda_n\PPP_{{\cal T}(\widehat\UU_{\OS_2})^\perp}\III^*\left(\AAA\left[ F_{\UU_1^*}^{-1}(\sign(\OO(\OOO^*)),\gamma\UU_1^*(\VV_1^*)^\top)\right]\right) + o_p(\lambda_n) \\
		&= \lambda_n\PPP_{{\cal T}(\UU_1^*)^\perp}\III^*\left(\AAA\left[ F_{\UU_1^*}^{-1}(\sign(\OO(\OOO^*)),\gamma\UU_1^*(\VV_1^*)^\top)\right]\right) + o_p(\lambda_n),
	\end{align*}
	where the last equality is due to the Lipchitz continuity of the projection operator and the fact that $\|\widehat\UU_{\OS_2} - \UU_1^*\|_{\max} <_P \lambda_n^{1-\eta}$. Therefore, by Assumption~\ref{ass:incor}, we have
	$$
	\begin{aligned}
		&g_{\gamma}\left(\PPP_{{{\cal S}(\OOO^*)}^\perp}\nabla l(\widehat\OOO_{\OS_2},\widehat\LL_{\OS_2}),  \PPP_{{\cal T}(\widehat\UU_{\OS_2})^\perp}\nabla l(\widehat\OOO_{\OS_2},\widehat\LL_{\OS_2})\right) \\
		=& \lambda_n g_{\gamma}(F_{\UU_1^*}^\perp F_{\UU_1^*}^{-1}(\sign(\OO(\OOO^*)),\gamma\UU_1^*(\VV_1^*)^\top) + o_p(\lambda_n).
	\end{aligned}
	$$ 
	Then the desired first order condition in \eqref{eq:orthogonal cond} holds in probability following Lemma 6 in \cite{chen2016fused},  which means $(\widehat\OOO_{\OS_2},\widehat\LL_{\OS_2})$ is a minimizer of $h(\OOO,\LL)$ in $\OS_1$. Its uniqueness can be shown following similar argument as for Lemma 7 in \cite{chen2016fused}.  
	
	Therefore, $h(\OOO,\LL)$ has a unique minimizer in $\OS_1\subset\OS_0,$ denoted as $(\widehat\OOO,\widehat\LL),$ which is in the interior of $\OS_1$ and also belongs to $\OS_2.$  As $h(\OOO,\LL)$ is a convex function, $\OS_0$ is a convex set, and $\OS_1$ contains an open neighborhood of $(\widehat\OOO,\widehat\LL),$ we conclude that $(\widehat\OOO,\widehat\LL)$ is also the unique minimizer of $h(\OOO,\LL)$ in $\OS_0.$ The estimation consistencies in \eqref{eq:consis} also hold due to the fact that $\lambda_n^{1-2\eta} < n^{-1/2+2\eta}.$ Furthermore, by definition of $\OS_2$, we have $\widehat\OOO-\OOO^*\in\SSSS(\OOO^*)$ and $\rank(\widehat\LL)\leq K$. Then, the sign consistency of $\widehat\OOO$ follows since $\min\{|\omega_{ij}^*|:\omega_{ij}^*\neq0\} \geq c>0$ for a constant $c$, and the rank consistency of $\widehat\LL$ follows since $$
	\begin{aligned}
		\sigma_K(\widehat\LL) \geq \sigma_K(\LL^*) - & \|\widehat\LL - \LL^*\|_2 \geq \sigma_K(\LL^*) - \|\widehat\LL - \LL^*\|_F \\
		&\geq \sigma_K(\LL^*) - p\|\widehat\LL - \LL^*\|_{\max}>_P \frac{1}{2}\sigma_K(\LL^*) > 0. 
	\end{aligned}
	$$
	This completes the proof of Theorem \ref{thm:Omega}. \end{proof}

\begin{proposition}\label{prop:LayerOrder}
	Under the same conditions of Theorem \ref{thm:Omega}, there holds $ 
	\Pr \big(\widehat{\boldsymbol\pi} \in \Pi^*\big) \to 1 \ \text{as} \  n\to \infty$,
	where $\Pi^*$ is the set of all possible true causal orderings of the chain components.
\end{proposition}

\begin{proof}[\bf Proof of Proposition~\ref{prop:LayerOrder}]
The proof is deferred to the supplement.
\end{proof}


\begin{proof}[\bf Proof of Theorem \ref{thm:Graph}]
	Let ${\cal E}={\cal E}_u \cup {\cal E}_d$, where ${\cal E}_u$ is the set of undirected edges and ${\cal E}_d$ is the set of directed edges. Then, 
	\begin{align*}
		\Pr(\widehat{\cal G} \neq {\cal G}^*) &= \Pr(\{\widehat{\cal E}_u \neq {\cal E}_u^*\} \cup \{\widehat{\cal E}_d \neq {\cal E}_d^*\} \cup \{\widehat{\boldsymbol\pi} \notin \Pi^*\})\\
		& \leq \Pr (\widehat{\cal E}_u \neq {\cal E}_u^*) + \Pr (\widehat{\boldsymbol\pi} \notin \Pi^*, \widehat{\cal E}_u = {\cal E}_u^*) + \Pr(\widehat{\cal E}_d \neq {\cal E}_d^*,  \widehat{\cal E}_u = {\cal E}_u^*, \widehat{\boldsymbol\pi} \in \Pi^*)\\
		&\leq \Pr (\widehat{\cal E}_u \neq {\cal E}_u^*) + \Pr (\widehat{\boldsymbol\pi} \notin \Pi^* | \widehat{\cal E}_u = {\cal E}_u^*) + \Pr(\widehat{\cal E}_d \neq {\cal E}_d^* | \widehat{\cal E}_u = {\cal E}_u^*, \widehat{\boldsymbol\pi} \in \Pi^*).
	\end{align*}
	By Theorem \ref{thm:Omega} and Proposition \ref{prop:LayerOrder}, we have $\Pr (\widehat{\cal E}_u \neq {\cal E}_u^*) \to 0$ and $\Pr (\widehat{\boldsymbol\pi} \notin \Pi^* | \widehat{\cal E}_u = {\cal E}_u^*) \to 0$, respectively. It then suffices to bound $\Pr(\widehat{\cal E}_d \neq {\cal E}_d^* | \widehat{\cal E}_u = {\cal E}_u^*, \widehat{\boldsymbol\pi} \in \Pi^*)$. 
	
	Note that $\{\sign(\widehat{\BB}) = \sign(\BB^*)\} \subseteq  \{\widehat{\cal E}_d = {\cal E}_d^*\}$, where $\widehat{\BB}$ is obtained by truncating $\widehat{\BB}^{\text{svd}}$ with a thresholding value $\nu_n$. Therefore, we turn to bound  
	\begin{align} \label{eq::thm3::decomB}
		\|\widehat{\BB}^{\text{svd}} - \BB^* \|_{\max} 
		\leq \|\widehat{\BB}^{\text{reg}} - \BB^*\|_{\max}	+ \| \widehat{\BB}^{\text{svd}} -  \widehat{\BB}^{\text{reg}}\|_{\max}.
	\end{align}
	To bound the first term of \eqref{eq::thm3::decomB}, we assume the true causal ordering of the chain components is $\boldsymbol{\pi}^*=(\pi_1^*,...,\pi_{m^*}^*)$. Since   $\widehat{\BB}^{\text{reg}}_{\tau_{\pi_k^*}^* {\cal C}^*_{k-1}}=\widehat{\bfSigma}_{\tau_{\pi_k^*}^* {\cal C}^*_{k-1}} (\widehat{\bfSigma}_{{\cal C}^*_{k-1} {\cal C}^*_{k-1}})^{-1}$ and $\BB_{\tau_{\pi_k^*}^* {\cal C}^*_{k-1}}^* = {\bfSigma}^*_{\tau_{\pi_k^*}^* {\cal C}^*_{k-1}} ({\bfSigma}^*_{{\cal C}^*_{k-1} {\cal C}^*_{k-1}})^{-1}$, we have
	$$
	\|\widehat{\BB}^{\text{reg}} - \BB^*\|_{\max}\leq \max_{2 \le k \le m^*} \Big \| \widehat{\bfSigma}_{\tau_{\pi_k^*}^* {\cal C}^*_{k-1}} (\widehat{\bfSigma}_{{\cal C}^*_{k-1} {\cal C}^*_{k-1}})^{-1} - {\bfSigma}^*_{\tau_{\pi_k^*}^* {\cal C}^*_{k-1}} ({\bfSigma}^*_{{\cal C}^*_{k-1} {\cal C}^*_{k-1}})^{-1} \Big \|_{\max}.
	$$
	It follows from the triangle inequality that
	\begin{align*} 
		&\Big \| \widehat{\bfSigma}_{\tau_{\pi_k^*}^* {\cal C}^*_{k-1}} (\widehat{\bfSigma}_{{\cal C}^*_{k-1} {\cal C}^*_{k-1}})^{-1} - {\bfSigma}^*_{\tau_{\pi_k^*}^* {\cal C}^*_{k-1}} ({\bfSigma}^*_{{\cal C}^*_{k-1} {\cal C}^*_{k-1}})^{-1} \Big \|_{\max}\nonumber\\
		= & \Big \|  \widehat{\bfSigma}_{\tau_{\pi_k^*}^* {\cal C}^*_{k-1}} \Big( (\widehat{\bfSigma}_{{\cal C}^*_{k-1} {\cal C}^*_{k-1}})^{-1} - ({\bfSigma}^*_{{\cal C}^*_{k-1} {\cal C}^*_{k-1}})^{-1} \Big) + \Big(  \widehat{\bfSigma}_{\tau_{\pi_k^*}^* {\cal C}^*_{k-1}} -  {\bfSigma}^*_{\tau_{\pi_k^*}^* {\cal C}^*_{k-1}}\Big) ({\bfSigma}^*_{{\cal C}^*_{k-1} {\cal C}^*_{k-1}})^{-1} \Big\|_{\max}\nonumber\\
		\leq & \Big \|  \widehat{\bfSigma}_{\tau_{\pi_k^*}^* {\cal C}^*_{k-1}} \Big\|_2 \Big\| (\widehat{\bfSigma}_{{\cal C}^*_{k-1} {\cal C}^*_{k-1}})^{-1} - ({\bfSigma}^*_{{\cal C}^*_{k-1} {\cal C}^*_{k-1}})^{-1} \Big\|_2 + \Big\|  \widehat{\bfSigma}_{\tau_{\pi_k^*}^* {\cal C}^*_{k-1}} -  {\bfSigma}^*_{\tau_{\pi_k^*}^* {\cal C}^*_{k-1}}\Big\|_2 \Big \| ({\bfSigma}^*_{{\cal C}^*_{k-1} {\cal C}^*_{k-1}})^{-1} \Big\|_{2}.
	\end{align*}
	By Corollary 2.4.5 of \cite{Golub2013}, we have $\big \|  \widehat{\bfSigma}_{\tau_{\pi_k^*}^* {\cal C}^*_{k-1}} \big\|_2 - \big \|  {\bfSigma}^*_{\tau_{\pi_k^*}^* {\cal C}^*_{k-1}} \big\|_2 
	\leq   \big \|  \widehat{\bfSigma}_{\tau_{\pi_k^*}^* {\cal C}^*_{k-1}} - {\bfSigma}^*_{\tau_{\pi_k^*}^* {\cal C}^*_{k-1}} \big\|_2 
	\leq  \big \| \widehat{\bfSigma}_{{\cal C}^*_{k} {\cal C}^*_{k}} - {\bfSigma}^*_{{\cal C}^*_{k} {\cal C}^*_{k}} \big \|_2$. Also, it follows from Lemma 1 of \cite{Ravikumar2011} that 
	$$
	\| \widehat{\bfSigma}_{{\cal C}_{k-1}^* {\cal C}_{k-1}^*} - \bfSigma_{{\cal C}_{k-1}^* {\cal C}_{k-1}^*}^* \|_{\max} \leq \| \widehat{\bfSigma} - \bfSigma^* \|_{\max}\lesssim_P n^{-\frac{1}{2}+\eta}.
	$$
	Since $|{\cal C}_{k-1}^*|$ and $\Lambda_{\min}( \bfSigma_{{\cal C}_{k-1}^* {\cal C}_{k-1}^*}^*)$ are fixed, we have 
	$n^{-\frac{1}{2}+\eta} \leq \frac{\Lambda_{\min}( \bfSigma_{{\cal C}_{k-1}^* {\cal C}_{k-1}^*}^*)}{2|{\cal C}_{k-1}^*|}$ when $n$ is sufficiently large. It then follows from Lemma 7 that 
	$$
	\| (\widehat{\bfSigma}_{{\cal C}^*_{k-1}  {\cal C}^*_{k-1} })^{-1} - (\bfSigma^*_{{\cal C}^*_{k-1} {\cal C}^*_{k-1} })^{-1}\|_{\max}	\lesssim_P \frac{2|{\cal C}^*_{k-1}|}{\Lambda_{\min}^2 (\bfSigma_{{\cal C}_{k-1}^* {\cal C}_{k-1}^*}^*)}n^{-\frac{1}{2} + \eta}.
	$$
	
	Combing all these results, we have
	\begin{align}\label{eq::thm3:B1}
		& \|\widehat{\BB}^{\text{reg}} - \BB^*\|_{\max}\nonumber\\
		\leq & \max_{2 \le k \le m^*} \Big \{ \Big(\Big \|  {\bfSigma}^*_{\tau_{\pi_k^*}^* {\cal C}^*_{k-1}} \Big\|_2 +  \Big \| \widehat{\bfSigma}_{{\cal C}^*_{k} {\cal C}^*_{k}} - {\bfSigma}^*_{{\cal C}^*_{k} {\cal C}^*_{k}} \Big \|_2\Big) \Big\| (\widehat{\bfSigma}_{{\cal C}^*_{k-1} {\cal C}^*_{k-1}})^{-1} - ({\bfSigma}^*_{{\cal C}^*_{k-1} {\cal C}^*_{k-1}})^{-1} \Big\|_2\nonumber\\
		&+ \Big\|  \widehat{\bfSigma}_{\tau_{\pi_k^*}^* {\cal C}^*_{k-1}} -  {\bfSigma}^*_{\tau_{\pi_k^*}^* {\cal C}^*_{k-1}}\Big\|_2 \Big \| ({\bfSigma}^*_{{\cal C}^*_{k-1} {\cal C}^*_{k-1}})^{-1} \Big\|_{2} \Big \}\nonumber\\
		\leq & \max_{2 \le k \le m^*} \Big \{ \Big( \Lambda_{\max}(\bfSigma^*) +  \Big \| \widehat{\bfSigma}_{{\cal C}^*_{k} {\cal C}^*_{k}} - {\bfSigma}^*_{{\cal C}^*_{k} {\cal C}^*_{k}} \Big \|_2\Big) \Big\| (\widehat{\bfSigma}_{{\cal C}^*_{k-1} {\cal C}^*_{k-1}})^{-1} - ({\bfSigma}^*_{{\cal C}^*_{k-1} {\cal C}^*_{k-1}})^{-1} \Big\|_2\nonumber\\
		&+ \Big\|  \widehat{\bfSigma}_{\tau_{\pi_k^*}^* {\cal C}^*_{k-1}} -  {\bfSigma}^*_{\tau_{\pi_k^*}^* {\cal C}^*_{k-1}}\Big\|_2 \frac{1}{\Lambda_{\min}(\bfSigma^*)} \Big \}\nonumber\\
		\lesssim&_P \max_{vm^*} \Big \{ \Big( \Lambda_{\max}(\bfSigma^*) + n^{-\frac{1}{2}+\eta} \Big) \frac{2|{\cal C}^*_{k-1}|}{\Lambda_{\min}^2 (\bfSigma_{{\cal C}_{k-1}^* {\cal C}_{k-1}^*}^*)}n^{-\frac{1}{2} + \eta} + \frac{1}{\Lambda_{\min}(\bfSigma^*)} n^{-\frac{1}{2}+\eta}  \Big \}\nonumber\\
		\leq &  \Big ( \frac{2\Lambda_{\max}(\bfSigma^*)p}{\Lambda_{\min}^2(\bfSigma^*)} + \frac{1}{\Lambda_{\min}(\bfSigma^*)} + \frac{2p}{\Lambda_{\min}^2(\bfSigma^*)} n^{-\frac{1}{2}+\eta} \Big ) n^{-\frac{1}{2}+\eta} \lesssim_P \ n^{-\frac{1}{2} + \eta}.
	\end{align}
	
	To bound the second term of \eqref{eq::thm3::decomB}, 
	we have
	\begin{align*}\label{eq::thm3::B2}
		\| \widehat{\BB}^{\text{svd}} - \widehat{\BB}^{\text{reg}} \|_{\max} &= \| \widehat{\UU}^{\text{svd}} (\widehat{\DD}^{\text{svd}} - \widehat{\DD}^{\text{reg}} ) (\widehat{\VV}^{\text{svd}})^\top \|_{\max}\nonumber\\
		& \leq \| \widehat{\UU}^{\text{svd}} \|_2 \| \widehat{\DD}^{\text{svd}} - \widehat{\DD}^{\text{reg}} \|_2 \| \widehat{\VV}^{\text{svd}}\|_2 =  \| \widehat{\DD}^{\text{svd}} - \widehat{\DD}^{\text{reg}} \|_{\max}  \leq \kappa_n.
	\end{align*}
	Therefore, $\| \widehat{\BB}^{\text{svd}} - \BB^* \|_{\max} \lesssim_P n^{-\frac{1}{2}+\eta}$. Then, since $\nu_n = n^{-\frac{1}{2}+2\eta}$, we obtain $$
	\Pr(\sign(\widehat{\BB}) = \sign(\BB^*)\mid \widehat{\cal E}_u = {\cal E}_u^*, \widehat{\boldsymbol\pi} \in \Pi^*) \to 1.
	$$ We complete the proof since $\{\sign(\widehat{\BB}) = \sign(\BB^*)\} \subseteq  \{\widehat{\cal E}_d = {\cal E}_d^*\}$.
\end{proof}

\section*{Supplement: Auxiliary lemmas}

Similar to $F_{\UU_1},$ we denote $G_{\UU_1}:{{\cal S}(\OOO^*)}\times (D(\UU_1)-\LL^*)\to {{\cal S}(\OOO^*)}\times D(\UU_1)$ as 
$$
G_{\UU_1}(\DDD_{\OOO},\DDD_{\LL}) = (\PPP_{{{\cal S}(\OOO^*)}}\III^*v(\DDD_{\OOO}+\DDD_{\LL}),\PPP_{D(\UU_1)}\III^*v(\DDD_{\OOO}+\DDD_{\LL})).
$$ 

\begin{lemma}\label{lem:invertible}
	For $(\OOO^*,\BB^*)\in\PA,$ both $F_{\UU_1^*}$ and $G_{\UU_1^*}$ are invertible. Furthermore, there exists a small $\epsilon>0$ such that for any $\UU_1\in\R^{p\times K}$ satisfying $\|\UU_1-\UU_1^*\|_{\max}<\epsilon,$ it holds true that $
	\|F^{-1}_{\UU_1}\|\leq C$ and $\|G_{\UU_1}^{-1}\|\leq C,$ where the norm is the operator norm induced by $g_{\gamma}$.
\end{lemma}

\begin{proof}[Proof of Lemma~\ref{lem:invertible}]
	The invertibility of $F_{\UU_1^*}$ and $G_{\UU_1^*}$ is shown by Lemma 1 in \cite{chen2016fused}. Further, it is not difficult to show that for there exists a $\epsilon>0$ such that if $\|\UU_1-\UU_1^*\|_{\max}<\epsilon,$ then $\SSSS(\OOO^*)\cap\TT(\UU_1)=\emptyset,$ which implies the invertibility of $F_{\UU_1}$ and $G_{\UU_1}.$ The boundedness of their inverses follows. 
\end{proof}

\begin{lemma}\label{lem:norm ratio}
	For $(\OOO^*,\BB^*)\in\PA,$ it holds true that $$
	\inf_{\substack{(\OOO,\LL) \in {{\cal S}(\OOO^*)}\times {\cal T}(\LL^*) \\ \LL \neq \textbf{0}}} \frac{\|\OOO+\LL\|_{\max}}{\|\LL\|_{\max}} >0.
	$$
\end{lemma}

\begin{proof}[Proof of Lemma~\ref{lem:norm ratio}]
	The proof of Lemma~\ref{lem:norm ratio} is similar to Lemma 3 in \cite{chen2016fused}, and thus omitted here.
\end{proof}

\begin{lemma}\label{lem::identiOrder}
	Suppose that $\bx = (x_1,...,x_p)^\top\in \mathbb{R}^p$ is generated from the linear SEM model \eqref{eq:model} with parameters $(\OOO,\BB)\in \PA$. For any $0\leq s \leq m-2$ and $k \in [m] \setminus  \cup_{j=1}^s \pi_j$, it holds true that
	\begin{equation}
		{\cal D}(\tau_k, {\cal C}_s) = \max_{i\in \tau_k} \Big ( \Var ({x}_{i} | \mathbf{x}_{{\cal C}_s}) - \OOO_{ii}^{-1} \Big ) \left \{
		\begin{aligned}
			= & 0, & ~~\text{if}~~\pa(\tau_k) \subseteq {\cal C}_s, \\
			> & 0, & ~~\text{if}~~ \pa(\tau_k) \nsubseteq {\cal C}_s,
		\end{aligned}
		\right.
	\end{equation}
	where ${\cal C}_0 = \emptyset$ and ${\cal C}_s = \cup_{j=1}^s \tau_{\pi_j}$.
\end{lemma}

\begin{proof}[Proof of Lemma \ref{lem::identiOrder}]

It follows from \eqref{eq:model} that 
\begin{align}\label{eq::lem3::cov}
\Cov(\XX) = \SG = \MM \MM^\top,
\end{align}
where $\MM = \A \OOO^{-1/2}$ and $\A = (\II_p - \BB)^{-1} = (a_{ij})_{p\times p}$. Note that the off-diagonal element of $\A$ is called the total causal effect \cite{ChenW2019} from one node to another. Particularly, $a_{ij}$ is the sum of the effects of all directed paths from node $j$ to node $i$, where the effect of each directed path is the product of the coefficient of all directed edges, and $a_{ii} = 1$ for $i\in [p]$. Then,  it follows from \eqref{eq::lem3::cov} that the submatrix of $\SG$ corresponding to the node set $ \CCC_s \cup \tau_k$ is 
\begin{align}\label{eq::lem3::subcov}
\SG_{ \CCC_s \cup \tau_k,  \CCC_s \cup \tau_k} = \MM_{ \CCC_s \cup \tau_k, \CCC_s\cup \CCC_s^c} \MM_{ \CCC_s \cup \tau_k, \CCC_s\cup \CCC_s^c}^\top,
\end{align}
where $\CCC_s^c = [p]\backslash \CCC_s$ and
\begin{equation*}
	\MM_{ \CCC_s \cup \tau_k, \CCC_s\cup \CCC_s^c}  = 
	\begin{pmatrix}
		{\MM}_{\CCC_s,\CCC_s} & 	{\MM}_{\CCC_s,\CCC_s^c} \\
		{\MM}_{\tau_k, \CCC_s} & {\MM}_{\tau_k, \CCC_s^c} \\
	\end{pmatrix}.
\end{equation*}
Note that 
\begin{align}\label{eq::lem3::AOOO}
		\MM_{ \CCC_s \cup \tau_k, \CCC_s\cup \CCC_s^c}  = \A_{\CCC_s \cup \tau_k, \CCC_s\cup \CCC_s^c} \OOO^{-1/2}_{\CCC_s\cup \CCC_s^c, \CCC_s\cup \CCC_s^c},
\end{align}
where 
\begin{equation*}
	\A_{ \CCC_s \cup \tau_k, \CCC_s\cup \CCC_s^c}  = 
	\begin{pmatrix}
		{\A}_{\CCC_s,\CCC_s} & 	{\A}_{\CCC_s,\CCC_s^c} \\
		{\A}_{\tau_k, \CCC_s} & {\A}_{\tau_k, \CCC_s^c} \\
	\end{pmatrix}
\end{equation*}
and
\begin{equation*}
	\OOO^{-1/2}_{ \CCC_s \cup \CCC_s^c, \CCC_s\cup \CCC_s^c}  = 
	\begin{pmatrix}
	\OOO^{-1/2}_{\CCC_s,\CCC_s} & \OOO^{-1/2}_{\CCC_s,\CCC_s^c} \\
	\OOO^{-1/2}_{\CCC_s^c, \CCC_s} & \OOO^{-1/2}_{\CCC_s^c, \CCC_s^c} \\
	\end{pmatrix}.
\end{equation*}
It follows from the fact that there exists no directed edge from node in lower chain components to nodes in upper chain components that $\A_{\CCC_s,\CCC_s^c}={\bf 0}_{|\CCC_s|\times |\CCC_s^c|}$. Further, since there is no undirected edge across different chain components,  we have $\OOO^{-1/2}_{\CCC_s,\CCC_s^c} = (\OOO^{-1/2}_{\CCC_s^c, \CCC_s} )^\top = {\bf 0}_{|\CCC_s|\times |\CCC_s^c|}$. Thus, there holds
\begin{align*}
{\MM}_{\CCC_s,\CCC_s^c} = 	{\A}_{\CCC_s,\CCC_s} \OOO^{-1/2}_{\CCC_s,\CCC_s^c}  + {\A}_{\CCC_s,\CCC_s^c} \OOO^{-1/2}_{\CCC_s^c, \CCC_s^c} = {\bf 0}_{|\CCC_s|\times |\CCC_s^c|}.
\end{align*}
Then, we have 
\begin{align}
\Cov (\xx_{\tau_k} | \xx_{\CCC_s}) &= \SG_{\tau_k, \tau_k} - \SG_{\tau_k, \CCC_s}(\SG_{\CCC_s, \CCC_s})^{-1} \SG_{\CCC_s, \tau_k}\nonumber\\
&= \MM_{\tau_k, \CCC_s} \MM_{\tau_k, \CCC_s}^\top +  \MM_{\tau_k, \CCC_s^c} \MM_{\tau_k, \CCC_s^c}^\top -  \MM_{\tau_k, \CCC_s} \MM_{\CCC_s, \CCC_s}^\top ( \MM_{\CCC_s, \CCC_s} \MM_{\CCC_s, \CCC_s}^\top)^{-1}  \MM_{\CCC_s, \CCC_s} \MM_{\tau_k, \CCC_s}^\top \nonumber\\
& =  \MM_{\tau_k, \CCC_s^c} \MM_{\tau_k, \CCC_s^c}^\top \nonumber\\
& = \A_{\tau_k, \CCC_s^c} \OOO^{-1}_{\CCC_s^c, \CCC_s^c} \A_{\tau_k, \CCC_s^c}^\top \nonumber\\
& = \OOO_{\tau_k,\tau_k}^{-1} + \sum_{j\in[m] \backslash ((\cup_{j=1}^s \pi_j) \cup \{k\})}  \A_{\tau_k, \tau_j} \OOO^{-1}_{\tau_j, \tau_j} \A_{\tau_k, \tau_j}^\top, \label{eq::condi}
\end{align}
where the second equality follows from \eqref{eq::lem3::subcov} and the fourth equality follows from \eqref{eq::lem3::AOOO}.

If $\pa(\tau_k) \subseteq {\cal C}_s$, it follows from the fact that $\A_{\tau_k, \tau_j} = {\bf 0}_{|\tau_k|\times |\tau_j|}$ for $j\in[m] \backslash ((\cup_{j=1}^s \pi_j) \cup \{k\})$ that \eqref{eq::condi} can be rewritten as 
$\Cov (\xx_{\tau_k} | \xx_{{\cal C}_s})= \OOO_{\tau_k,\tau_k}^{-1}$, which implies that ${\cal D}(\tau_k, {\cal C}_s) = 0$.

If $\pa(\tau_k) \nsubseteq {\cal C}_s$, it follows from no semi-directed cycle assumption that there exists a node $l\in \tau_k$ and  another node $t\in \tau_i \subseteq \CCC_s^c \backslash \tau_k$ such that $t\in \pa (l)$ and no other directed path from node $t$ to node $l$. Then, we have 
\begin{align}
\Var (x_l | \xx_{\CCC_s}) &= e_l^\top \Cov (\xx_{\tau_k} | \xx_{\CCC_s}) e_l \nonumber\\
&= \OOO_{ll}^{-1} +   \sum_{j\in[m] \backslash ((\cup_{j=1}^s \pi_j) \cup \{k\})}  e_l^\top\A_{\tau_k, \tau_j} \OOO^{-1}_{\tau_j, \tau_j} \A_{\tau_k, \tau_j}^\top e_l\nonumber\\
& \geq \OOO_{ll}^{-1} +  (\A_{\tau_k, \tau_i}^\top e_l)^\top \OOO^{-1}_{\tau_i, \tau_i} \A_{\tau_k, \tau_i}^\top e_l, \label{eq::lem3::sandwich}
\end{align}
where $e_l\in\{0,1\}^{|\tau_k|}$ and only the element of $e_l$ corresponding to node $l$ is equal to $1$, and the last inequality follows from the fact that $\OOO_{\tau_j,\tau_j}^{-1}$ is positive definite. Next, denote $[\A_{\tau_k, \tau_i}]_{lt}$ as the element of $\A_{\tau_k, \tau_i}$ corresponding to node $l\in\tau_k$ and $t\in \tau_i$. Since $t\in \pa (l)$ and there is no other directed path from node $t$ to node $l$, we have $[\A_{\tau_k, \tau_i}]_{lt} = \beta_{lt}\neq 0$ which implies $\A_{\tau_k, \tau_i}^\top e_l \neq {\bf 0}_{|\tau_i|}$. Then, it follows from \eqref{eq::lem3::sandwich} and the fact that $\OOO^{-1}_{\tau_i, \tau_i}$ is positive definite that 
\begin{align*}
\Var (x_l | \xx_{\CCC_s}) > \OOO_{ll}^{-1}.
\end{align*}
Thus, there holds
\begin{align*}
		{\cal D}(\tau_k, {\cal C}_s) =  \max_{i\in \tau_k} \Big\{ \Var (x_{i} | \xx_{{\cal C}_s}) - \OOO_{ii}^{-1} \Big\}  \geq \Var (x_{l} | \xx_{{\cal C}_s}) - \OOO_{ll}^{-1} > 0.
\end{align*}
This completes the proof. \end{proof}


\begin{lemma}\label{lem:hat Theta}
	Let $\widehat\OS = {{\cal S}(\OOO^*)} \times D(\widehat\UU_{\OS_2}).$
	Then, with probability approaching 1, $h(\OOO,\LL)$ has a unique minimizer in $\widehat\OS$, denoted as $(\widehat\OOO_{\widehat\OS},\widehat\LL_{\widehat\OS})$. Let $\widehat\LL_{\widehat\OS} = \widehat\UU_{\OS_2}\widehat\DD_{\widehat\OS}\widehat\UU_{\OS_2}^\top$ be the eigen decomposition, then we have 
	$$
	\|\widehat\DD_{\widehat\OS} - \DD_1^*\|_{\max} \lesssim_P \lambda_n^{1-\eta}, \mbox{ and }
	\|\widehat\OOO_{\widehat\OS} - \OOO^*\|_{\max} \lesssim_P \lambda_n^{1-\eta}.
	$$
\end{lemma}

\begin{proof}[Proof of Lemma~\ref{lem:hat Theta}]
	We first show the existence and uniqueness for the minimizer of $h(\OOO,\LL)$ in $\widehat \OS = {{\cal S}(\OOO^*)} \times D(\widehat\UU_{\OS_2}).$ This is equivalent to show that there exists a unique $(\OOO,\LL) \in \widehat\OS$ such that the following first order condition holds: 
	\begin{equation}\label{eq:1st cond}
		\0 \in \PPP_{{{\cal S}(\OOO^*)}}\partial_{\OOO} h(\OOO,\LL), \text{~~and~~} \0 \in \PPP_{D(\widehat\UU_{\OS_2})}\partial_{\LL} h(\OOO,\LL).
	\end{equation}
	Since $h(\OOO,\LL)$ is convex on the linear subspace $\widehat\OS,$ it suffices to show there exists a unique $(\OOO,\LL)\in \BBB$ such that \eqref{eq:1st cond} holds, where 
	$$
	\BBB = \{ (\OOO,\LL) \in \widehat\OS: g_{\gamma}(\OOO-\OOO^*,\LL-\LL^*) \leq \lambda_n^{1-\eta}\} 
	$$ 
	is a neighborhood of $(\OOO^*,\LL^*)$ in $\widehat\OS.$
	Note that for $(\OOO,\LL) \in\BBB,$ we have 
	\begin{equation*}
		\PPP_{{{\cal S}(\OOO^*)}}\partial_{\OOO} \|\OOO\|_{1,\off} = \sign(\OO(\OOO^*))~~\text{and}~~\PPP_{D(\widehat\UU_{\OS_2})}\partial_{\LL} \|\LL\|_* = \widehat\UU_{\OS_2}\widehat\VV_{\OS_2}^\top,
	\end{equation*}
	where $\widehat\VV_{\OS_2}$ contains the right singular vectors of $\widehat\LL_{\OS_2},$ which is the same as $\widehat\UU_{\OS_2}$ except that the columns corresponding to the negative eigenvalues are multiplied by a negative sign. 
	
	In the following, we denote $\DDD_{\OOO} = \OOO-\OOO^*,~\DDD_{\LL} = \LL-\LL^*$ and $\DDD = \DDD_{\OOO} + \DDD_{\LL}.$ By \eqref{eq:taylor1}, we have $$
	\PPP_{{{\cal S}(\OOO^*)}}\partial_{\OOO} h(\OOO,\LL) = \PPP_{{{\cal S}(\OOO^*)}} \III^*v(\DDD_{\OOO}+\DDD_{\LL}) +\PPP_{{{\cal S}(\OOO^*)}} \nabla R_n(\DDD_{\OOO}+\DDD_{\LL}) + \lambda_n \sign(\OO(\OOO^*))
	$$ $$
	\PPP_{D(\widehat\UU_{\OS_2})}\partial_{\LL} h(\OOO,\LL) = \PPP_{D(\widehat\UU_{\OS_2})} \III^*v(\DDD_{\OOO}+\DDD_{\LL}) +\PPP_{D(\widehat\UU_{\OS_2})} \nabla R_n(\DDD_{\OOO}+\DDD_{\LL}) + \gamma\lambda_n \widehat\UU_{\OS_2}\widehat\VV_{\OS_2}^\top.
	$$ 
	Then, the first order condition becomes $$
	G_{\widehat\UU_{\OS_2}}(\DDD_{\OOO},\DDD_{\LL}) = - \left(\PPP_{{{\cal S}(\OOO^*)}} \nabla R_n(\DDD)+\lambda_n \sign(\OO(\OOO^*)),\PPP_{D(\widehat\UU_{\OS_2})} \nabla R_n(\DDD)+\gamma\lambda_n \widehat\UU_{\OS_2}\widehat\VV_{\OS_2}^\top\right).
	$$ 
	By Lemma~\ref{lem:invertible}, the mapping $G_{\widehat\UU_{\OS_2}}$ is invertible and $G^{-1}_{\widehat\UU_{\OS_2}}$ is bounded. Define an operator on $\BBB \setminus (\OOO^*,\LL^*)$ as 
	$$
	\CC_{\widehat\UU_{\OS_2}}(\DDD_{\OOO},\DDD_{\LL}) = -G^{-1}_{\widehat\UU_{\OS_2}}\left(\PPP_{{{\cal S}(\OOO^*)}} \nabla R_n(\DDD)+\lambda_n \sign(\OO(\OOO^*)),\PPP_{D(\widehat\UU_{\OS_2})} \nabla R_n(\DDD)+\gamma\lambda_n \widehat\UU_{\OS_2}\widehat\VV_{\OS_2}^\top\right).
	$$ 
	
	Lemma \ref{lem:contraction} shows that, with probability approaching 1, $\CC_{\widehat\UU_{\OS_2}}$ is a contraction mapping on $\BBB.$ 
	By the fixed point theorem, there exists a unique solution $(\widehat\DDD_{\OOO},\widehat\DDD_{\LL})\in\BBB$ such that 
	$$
	\CC_{\widehat\UU_{\OS_2}}(\widehat\DDD_{\OOO},\widehat\DDD_{\LL}) = (\widehat\DDD_{\OOO},\widehat\DDD_{\LL}).
	$$ We complete the proof by letting $\widehat\OOO_{\widehat\OS} = \OOO^*+\widehat\DDD_{\OOO}$ and $\widehat\LL_{\widehat\OS} = \LL^*+\widehat\DDD_{\LL}.$
\end{proof}

\begin{lemma}\label{lem:contraction}
	With probability approaching 1, $\CC_{\widehat\UU_{\OS_2}}(\BBB)\subset\BBB.$ Furthermore, $\CC_{\widehat\UU_{\OS_2}}$ is a contraction mapping on $\BBB.$
\end{lemma}

\begin{proof}[Proof of Lemma~\ref{lem:contraction}]
	For $(\DDD_{\OOO},\DDD_{\LL})\in\BBB,$ denote $\DDD = \DDD_{\OOO}+\DDD_{\LL}.$ Then, we have $$
	\begin{aligned}
		&~g_{\gamma}\left(\CC_{\widehat\UU_{\OS_2}}(\DDD_{\OOO},\DDD_{\LL})\right)\\
		& \leq~ g_{\gamma}\left(G^{-1}_{\widehat\UU_{\OS_2}}\left(\PPP_{{{\cal S}(\OOO^*)}} \nabla R_n(\DDD),\PPP_{D(\widehat\UU_{\OS_2})} \nabla R_n(\DDD)\right)\right) + g_{\gamma}\left(G^{-1}_{\widehat\UU_{\OS_2}}\left(\lambda_n \sign(\OO(\OOO^*)),\gamma\lambda_n \widehat\UU_{\OS_2}\widehat\VV_{\OS_2}^\top\right)\right)\\
		&\lesssim g_{\gamma}\left(\PPP_{{{\cal S}(\OOO^*)}} \nabla R_n(\DDD),\PPP_{D(\widehat\UU_{\OS_2})} \nabla R_n(\DDD)\right) + g_{\gamma}\left(\lambda_n \sign(\OO(\OOO^*)),\gamma\lambda_n \widehat\UU_{\OS_2}\widehat\VV_{\OS_2}^\top\right),
	\end{aligned}
	$$ where the second inequality is due to Lemma~\ref{lem:invertible}. By \eqref{eq:remainder2}, we have $$
	g_{\gamma}\left(\PPP_{{{\cal S}(\OOO^*)}} \nabla R_n(\DDD),\PPP_{D(\widehat\UU_{\OS_2})} \nabla R_n(\DDD)\right) = O_p\left(\frac{1}{\sqrt{n}}\right)+O(\lambda_n^{2(1-\eta)}) \lesssim_P \lambda_n^{1-\eta},
	$$ and it hold true that $
	g_{\gamma}\left(\lambda_n \sign(\OO(\OOO^*)),\gamma\lambda_n \widehat\UU_{\OS_2}\widehat\VV_{\OS_2}^\top\right) = \lambda_n \leq \lambda_n^{1-\eta}.$ Therefore, $\CC_{\widehat\UU_{\OS_2}}(\DDD_{\OOO},\DDD_{\LL}) \in \BBB.$ 
	
	For any other $(\widetilde\DDD_{\OOO},\widetilde\DDD_{\LL})\in\BBB,$ denote $\widetilde\DDD = \widetilde\DDD_{\OOO}+\widetilde\DDD_{\LL},$ we have $$
	\begin{aligned}
		&g_{\gamma}\left(\CC_{\widehat\UU_{\OS_2}}(\DDD_{\OOO},\DDD_{\LL})-\CC_{\widehat\UU_{\OS_2}}(\widetilde\DDD_{\OOO},\widetilde\DDD_{\LL})\right) \\
		&=g_{\gamma}\left(\PPP_{{{\cal S}(\OOO^*)}} \left[\nabla R_n(\DDD)-\nabla R_n(\widetilde\DDD)\right],\PPP_{D(\widehat\UU_{\OS_2})} \left[\nabla R_n(\DDD) - \nabla R_n(\widetilde\DDD)\right]\right)\\
		&=g_{\gamma}\left(\PPP_{{{\cal S}(\OOO^*)}} \left[\nabla R_n^2(\check\DDD)v(\DDD-\widetilde\DDD)\right],\PPP_{D(\widehat\UU_{\OS_2})} \left[\nabla R_n^2(\check\DDD)v(\DDD-\widetilde\DDD)\right]\right)\\
		&\lesssim \|\nabla R_n^2(\check\DDD)\|_{\max}\|\DDD-\widetilde\DDD\|_{\max}\\
		&\lesssim \|\nabla R_n^2(\check\DDD)\|_{\max}g_{\gamma}(\DDD_{\OOO}-\widetilde\DDD_{\OOO},\DDD_{\LL}-\widetilde\DDD_{\LL}),
	\end{aligned}
	$$ 
	where the second equality is due to Taylor's expansion. According to \eqref{eq:remainder3}, with probability approaching 1, $\CC_{\widehat\UU_{\OS_2}}$ is a contraction mapping on $\BBB.$ \end{proof}

\begin{lemma}\label{lem:tmp mat}
	Let $\widehat\DDD_{\LL} = \UU_1^*\DD_1^*(\widehat\UU_{\OS_2} - \UU_1^*)^\top+(\widehat\UU_{\OS_2}-\UU_1^*)\DD_1^*(\UU_1^*)^\top + \UU_1^*(\widehat\DD_{\widehat\OS} -\DD_1^*)(\UU_1^*)^\top$, then the followings hold:
	\begin{enumerate}
		\item $\widehat\DDD_{\LL} \in {\cal T}(\LL^*);$
		\item $\|\widehat\LL_{\widehat\OS} - \LL^* - \widehat\DDD_{\LL}\|_{\max} \lesssim_P \lambda_n^{2(1-\eta)};$
		\item there exists a constant $c>0$ such that
		$\|\widehat\DDD_{\LL}\|_{\max} >_P c\|\widehat\UU_{\OS_2} - \UU_1^*\|_{\max} - c\lambda_n^{2(1-\eta)}.$ 
	\end{enumerate}
\end{lemma}

\begin{proof}[Proof of Lemma~\ref{lem:tmp mat}]
	The first two results follow directly from Lemma 4 in \cite{chen2016fused}. For the third one, by the Davis-Kahan theorem \cite[][]{yu2015useful}, we obtain 
	$$
	\|\widehat\UU_{\OS_2} - \UU_1^*\|_2 \lesssim \frac{\|\widehat\LL_{\widehat\OS} - \LL^*\|_F}{\min_{1\leq k\leq K}\{\lambda_{k}(\LL^*)-\lambda_{k+1}(\LL^*)\}},
	$$ 
	where $\lambda_k(\LL^*)$ is the $k$-th largest eigenvalues of $\LL^*.$
	By Assumption~\ref{ass:diff}, it implies that $\|\widehat\LL_{\widehat\OS} - \LL^*\|_{\max} > c\|\widehat\UU_{\OS_2} - \UU_1^*\|_{\max}$ for a constant $c>0.$ Therefore, we obtain $\|\widehat\DDD_{\LL}\|_{\max} > \|\widehat\LL_{\widehat\OS} - \LL^*\|_{\max} - \|\widehat\LL_{\widehat\OS} - \LL^* - \widehat\DDD_{\LL}\|_{\max} >_P c\|\widehat\UU_{\OS_2} - \UU_1^*\|_{\max} - c\lambda_n^{2(1-\eta)}$, which completes the proof. \end{proof}

\begin{lemma}\label{lem::MatrixInverse}
	Suppose $\mathbf{A}\in \mathbb{R}^{p\times p}$ is a positive definite matrix with $\Lambda_{\min}(\mathbf{A}) \ge 2p \epsilon$, and $\mathbf{E}\in \mathbb{R}^{p\times p}$ is a symmetric error matrix with $\| \mathbf{E} \|_{\max} \leq \epsilon$, then $\mathbf{A}+ \mathbf{E}$ is invertible and 
	$$\| (\mathbf{A}+ \mathbf{E})^{-1} - \mathbf{A}^{-1} \|_{\max}\leq \frac{2p}{\Lambda_{\min}^2(\mathbf{A})}\epsilon.$$	
\end{lemma}
\begin{proof}[Proof of Lemma~\ref{lem::MatrixInverse}]
	Note that Lemma \ref{lem::MatrixInverse} modifies the results in Lemma 5 in \cite{Harris2013} and Lemma 29 in \cite{Loh2014}.
	First, to prove that $\mathbf{A} + \mathbf{E}$ is inverible, we note that $\mathbf{A} + \mathbf{E} = (\mathbf{I} + \mathbf{E}\mathbf{A}^{-1}) \mathbf{A}$, and
	\begin{align}\label{eq::lem7::eq1}
		\|\mathbf{E}\mathbf{A}^{-1}\|_2 \leq \|\mathbf{E}\|_2 \|\mathbf{A}^{-1}\|_2 \leq p \|\mathbf{E}\|_{\max} \|\mathbf{A}^{-1}\|_2 \leq \frac{p}{\Lambda_{\min}(\mathbf{A})} \epsilon \leq \frac{1}{2},
	\end{align}
	which implies $\mathbf{I} + \mathbf{E}\mathbf{A}^{-1}$ is invertible, and thus $\mathbf{A} + \mathbf{E}$ is also invertible. 
	
	Moreover, we have
	\begin{align*}
		\| (\mathbf{A}+ \mathbf{E})^{-1} - \mathbf{A}^{-1} \|_{\max} &\leq 	\| (\mathbf{A}+ \mathbf{E})^{-1} - \mathbf{A}^{-1} \|_{2} = \| (\mathbf{A}+\mathbf{E})^{-1}(\mathbf{A}+\mathbf{E} - \mathbf{A}) \mathbf{A}^{-1}\|_2 \\
		& \leq  \| (\mathbf{A}+\mathbf{E})^{-1}\|_2 \|\mathbf{E}\mathbf{A}^{-1}\|_2 \leq \frac{\|\mathbf{A}^{-1}\|_2}{1-\|\mathbf{E}\mathbf{A}^{-1}\|_2} \|\mathbf{E}\mathbf{A}^{-1}\|_2 \\
		& \leq 2 \|\mathbf{A}^{-1}\|_2  \|\mathbf{E}\mathbf{A}^{-1}\|_2 \leq \frac{2p}{\Lambda_{\min}^2(\mathbf{A})}\epsilon,
	\end{align*}
	where the third inequality follows the inequality (5.8.2) in \cite{Horn2012}, and the last inequality follows from \eqref{eq::lem7::eq1}.
	
\end{proof}

\section*{Supplement: Proof of Propositions~\ref{prop:ident} and \ref{prop:LayerOrder}}

\begin{proof}[\bf Proof of Propositions~\ref{prop:ident}]
	First of all, it can be verified that $\partial_{\OOO} \bar l(\OOO^*,\LL^*) = \partial_{\LL} \bar l(\OOO^*,\LL^*) = {\bf0}.$ By the concavity of $\bar l(\OOO^*,\LL^*),$ this implies that $(\OOO^*,\LL^*)$ is a maximizer of $\bar l(\OOO^*,\LL^*).$
	
	For any small $\epsilon>0$, let $(\widetilde\OOO,\widetilde\LL)\in\OS(\epsilon)$ be a maximizer of $\bar l(\OOO,\LL)$ in the interior of $\OS(\epsilon)$. Then, the first order conditions hold: 
	\begin{equation}\label{eq:order pop}
		\PPP_{{{\cal S}(\OOO^*)}}\partial_{\OOO} \bar l(\widetilde\OOO,\widetilde\LL) = {\bf0}, ~~\text{and}~~\PPP_{{\cal T}(\widetilde\LL)}\partial_{\LL} \bar l(\widetilde\OOO,\widetilde\LL) = {\bf0}.
	\end{equation} 
	Let $\widetilde\LL = \widetilde\UU_1\widetilde\DD_1\widetilde\UU_1^\top$ be the eigen decomposition of $\widetilde\LL.$
	Define $\widetilde\DDD_{\OOO} = \widetilde\OOO-\OOO^*,\widetilde\DDD_{\LL} = \widetilde\LL-\LL^*$ and $\widetilde\DDD = \widetilde\DDD_{\OOO} + \widetilde\DDD_{\LL}.$ 
	Further let $\widetilde\DDD_{\LL} = \widetilde\DDD_{\LL,1} + \widetilde\DDD_{\LL,2},$ where 
	$$
	\widetilde\DDD_{\LL,1} = 
	\widetilde\UU_1\widetilde\DD_1(\widetilde\UU_1 - \UU_1^*)^\top+(\widetilde\UU_1-\UU_1^*)\widetilde\DD_1\widetilde\UU_1^\top + \widetilde\UU_1(\widetilde\DD_1 -\DD_1^*)\widetilde\UU_1^\top.
	$$ 
	It can be verified that $\widetilde\DDD_{\LL,1}\in\TT(\widetilde\LL)$ and $\|\widetilde\DDD_{\LL,2}\|_{\max}\leq C\epsilon^2$ for some constant $C$ \citep[][Lemma~4]{chen2016fused}.
	
	By Taylor's expansion, we have 
	\begin{equation}\label{eq:tayler pop}
		\begin{aligned}
			\bar l(\TTT^*+\widetilde\DDD) &= \bar l(\TTT^*) - \frac{1}{2}v(\widetilde\DDD)^\top\III^*v(\widetilde\DDD) + R(\widetilde\DDD),
		\end{aligned}
	\end{equation}
	where it can be verified that $\| R(\widetilde\DDD)\|_{\max} = O(\|\widetilde\DDD\|_{\max}^3)$ and $\|\nabla R(\widetilde\DDD)\|_{\max} = O(\|\widetilde\DDD\|_{\max}^2).$ Then, by \eqref{eq:order pop} and \eqref{eq:tayler pop}, we obtain $$
	\begin{aligned}
		{\bf0} = -\PPP_{\SSSS(\OOO^*)}\nabla\bar l(\TTT^*+\widetilde\DDD) =& \PPP_{\SSSS(\OOO^*)}(\III^*v(\widetilde\DDD)) - \PPP_{\SSSS(\OOO^*)}(\nabla R(\widetilde\DDD)) \\
		{\bf0} = -\PPP_{{\cal T}(\widetilde\LL)}\nabla\bar l(\TTT^*+\widetilde\DDD) =& \PPP_{{\cal T}(\widetilde\LL)}(\III^*v(\widetilde\DDD)) - \PPP_{{\cal T}(\widetilde\LL)}(\nabla R(\widetilde\DDD)),
	\end{aligned}
	$$ which leads that $$
	F_{\widetilde\UU_1}(\widetilde\DDD_{\OOO},\widetilde\DDD_{\LL,1}) = \left( \PPP_{{{\cal S}(\OOO^*)}} \nabla R(\widetilde\DDD),\PPP_{{\cal T}(\widetilde\LL)} \left(\nabla R(\widetilde\DDD) - \III^*v(\widetilde\DDD_{\LL,2}) \right) \right).
	$$ 
	
	When $\epsilon$ is sufficiently small, Lemma 1 implies that $F_{\widetilde\UU_1}$ is invertible and $F_{\widetilde\UU_1}^{-1}$ is bounded. Then, we have 
	$$
	(\widetilde\DDD_{\OOO},\widetilde\DDD_{\LL,1}) = -F_{\widetilde\UU_1}^{-1}\left( \PPP_{{{\cal S}(\OOO^*)}} \nabla R(\widetilde\DDD),\PPP_{{\cal T}(\widetilde\LL)} \left(\nabla R(\widetilde\DDD) - \III^*v(\widetilde\DDD_{\LL,2}) \right) \right).
	$$ 
	By Lemma 1 and the fact that $\|\nabla R(\widetilde\DDD)\|_{\max} = O(\|\widetilde\DDD\|_{\max}^2) = O(\|(\widetilde\DDD_{\OOO},\widetilde\DDD_{\LL,1})\|_{\max}^2 + \|\widetilde\DDD_{\LL,2}\|_{\max}^2 )$, we have $\|(\widetilde\DDD_{\OOO},\widetilde\DDD_{\LL,1})\|_{\max} \leq C\left[\|(\widetilde\DDD_{\OOO},\widetilde\DDD_{\LL,1})\|_{\max}^2 + \|\widetilde\DDD_{\LL,2}\|_{\max} \right]$. Then, $$
	\|(\widetilde\DDD_{\OOO},\widetilde\DDD_{\LL,1})\|_{\max} \leq C\|\widetilde\DDD_{\LL,2}\|_{\max} \leq C\epsilon^2
	$$ 
	for possibly different constants $C$, which further implies that 
	$$
	\|(\widetilde\DDD_{\OOO},\widetilde\DDD_{\LL})\|_{\max} \leq \|(\widetilde\DDD_{\OOO},\widetilde\DDD_{\LL,1})\|_{\max} + \|\widetilde\DDD_{\LL,2}\|_{\max} \leq C\epsilon^2,
	$$ 
	and hence that $(\widetilde\OOO,\widetilde\LL)\in\OS(C\epsilon^2)$. Then, when $\epsilon$ is sufficiently small, we have $C\epsilon^2<\epsilon/2$, and thus $(\widetilde\OOO,\widetilde\LL)$ must be a maximizer of $\bar l(\OOO,\LL)$ in the interior of $\OS(\epsilon/2)$. 
	
	Repeating the above derivation implies that $(\widetilde\OOO,\widetilde\LL) = (\OOO^*,\LL^*)$. Therefore, $(\OOO^*,\LL^*)$ is the unique maximizer of $\bar l(\OOO^*,\LL^*)$ in the interior of $\OS(\epsilon)$, which implies that 
	it is the unique maximizer of $\bar l(\OOO^*,\LL^*)$ in $\OS(\epsilon/2)$. This completes the proof of Proposition~\ref{prop:ident}. 
	\end{proof}

	\begin{proof}[\bf Proof of Proposition~\ref{prop:LayerOrder}]

	We restrict the following analysis on the event $\{\widehat m = m^*\}\cap\{\widehat{\tau}_i = \tau_i^* : i\in [m^*]\}$, which occurs with probability approaching 1 according to Theorem~\ref{thm:Omega}.
	We establish the results by mathematical induction. Specifically, suppose we have correctly determined the first $k$ chain components among $\widehat \tau_i$'s, so that $(\widehat \pi_1,\ldots,\widehat \pi_k) \in \Pi_k^*$, where $\Pi_k^*$ denotes the set of all possible true causal ordering of the first $k$ chain components among $\tau_l^*$'s. Let $\widehat{\cal C}_k = \cup_{s=1}^k \widehat \tau_{\widehat \pi_s}$, and we consider two chain components $\widehat \tau_i$ and $\widehat \tau_j$ that are not subsets of $\widehat{\cal C}_k$, but $
	\pa(\widehat \tau_j) \subseteq \widehat{\cal C}_k \mbox{ and } \pa(\widehat \tau_i) \nsubseteq \widehat{\cal C}_k.$
	It then suffices to show that $\widehat{\cal D}(\widehat \tau_i, \widehat{\cal C}_k) - \widehat{\cal D}(\widehat \tau_j, \widehat{\cal C}_k)>0$.
	
	Simple algebra yields that
	\begin{align}\label{eq::lem7::decom}
		&\widehat{\cal D}(\widehat{\tau}_i, \widehat{\cal C}_k) - \widehat{\cal D}(\widehat{\tau}_j, \widehat{\cal C}_k) \nonumber\\
		=& \big({\cal D}(\widehat{\tau}_i, \widehat{\cal C}_k) - 
		{\cal D}(\widehat{\tau}_j, \widehat{\cal C}_k)\big) + \big( \widehat{\cal D}(\widehat{\tau}_i, \widehat{\cal C}_k) - {\cal D}(\widehat{\tau}_i, \widehat{\cal C}_k)\big ) + \big( {\cal D}(\widehat{\tau}_j, \widehat{\cal C}_k) - \widehat{\cal D}(\widehat{\tau}_j, \widehat{\cal C}_k)\big ) \nonumber \\
		\geq & \big({\cal D}(\widehat{\tau}_i, \widehat{\cal C}_k) - 
		{\cal D}(\widehat{\tau}_j, \widehat{\cal C}_k)\big) - 2 \max_{\widehat{\tau}_i \nsubseteq \widehat{\cal C}_k} \big| \widehat{\cal D}(\widehat{\tau}_i, \widehat{\cal C}_k) - {\cal D}(\widehat{\tau}_i, \widehat{\cal C}_k)\big |.
	\end{align}
	By Lemma 3, and facts $\widehat m = m^*$ and $\widehat{\tau}_i = \tau_i^*~\text{for}~i\in [m^*]\}$, we have
	${\cal D}(\widehat{\tau}_j, \widehat{\cal C}_k)=0$. Furthermore, there exists a node $l\in \widehat{\tau}_i$ such that $\pa(l)\backslash\widehat{\cal C}_k \neq \emptyset$, and thus
	$$
	{\cal D}(\widehat{\tau}_i, \widehat{\cal C}_k)\geq \sum_{t\in \pa(l)\backslash \widehat{\cal C}_k} \beta_{lt}^2 \Var(x_t | \xx_{\widehat{\cal C}_k})>0,
	$$
	where the first inequality follows from (3) in Lemma 3. Therefore, there exists a positive constant $c$ such that ${\cal D}(\widehat{\tau}_i, \widehat{\cal C}_k) - {\cal D}(\widehat{\tau}_j, \widehat{\cal C}_k)\geq c > 0.$
	

	For the second term in \eqref{eq::lem7::decom}, let $\bfSigma^*(l, \widehat{\cal C}_k)={\mathbf{\Sigma}}^*_{ll} -{\mathbf{\Sigma}}^*_{l \widehat{\CCC}_{k}} ({\mathbf{\Sigma}}^*_{\widehat{\CCC}_{k}\widehat{\CCC}_{k}})^{-1} {\mathbf{\Sigma}}^*_{ \widehat{\CCC}_{k}l}$ and $\widehat{\bfSigma}(l, \widehat{\cal C}_k) = \widehat{\mathbf{\Sigma}}_{ll} - \widehat{\mathbf{\Sigma}}_{l \widehat{\CCC}_{k}} \widehat{\mathbf{\Sigma}}_{\widehat{\CCC}_{k}\widehat{\CCC}_{k}}^{-1} \widehat{\mathbf{\Sigma}}_{ \widehat{\CCC}_{k}l}$. 
	It holds that
	\begin{align*} 
		\big | \widehat{\cal D}(\widehat{\tau}_i, \widehat{\cal C}_k) -  {\cal D}(\widehat{\tau}_i, \widehat{\cal C}_k)\big | 
		=& \big | \max_{l\in \widehat{\tau}_i}  \big\{ \widehat{\bfSigma}(l, \widehat{\cal C}_k) - \widehat{\bfOmega}_{ll}^{-1} \big\}  - \max_{l\in \widehat{\tau}_i}  \big\{ {\bfSigma}^*(l, \widehat{\cal C}_k) - ({\bfOmega}^*)_{ll}^{-1} \big\}\big |\nonumber\\
		\leq & \max_{l\in \widehat{\tau}_i} \big |  \big\{ \widehat{\bfSigma}(l, \widehat{\cal C}_k) - \widehat{\bfOmega}_{ll}^{-1} \big\}  - \big\{ {\bfSigma}^*(l, \widehat{\cal C}_k) - ({\bfOmega}^*)_{ll}^{-1} \big\}\big | \nonumber\\
		\leq & \max_{l\in \widehat{\tau}_i} \big |   \widehat{\bfSigma}(l, \widehat{\cal C}_k) - {\bfSigma}^*(l, \widehat{\cal C}_k) \big | +  \max_{l\in \widehat{\tau}_i} \big | \widehat{\bfOmega}_{ll}^{-1} -   ({\bfOmega}^*)_{ll}^{-1} \big |,
	\end{align*}
	which converges to 0 in probability by Theorem~\ref{thm:Omega} and $\widehat\SG\overset{pr}{\to}\SG^*$. Therefore, $$
	\widehat{\cal D}(\widehat{\tau}_i,\widehat{\cal C}_k) - \widehat{\cal D}(\widehat{\tau}_j,\widehat{\cal C}_k) >_P 0.
	$$
	As a direct consequence, the selected chain component has to be the one whose parent nodes are contained in $\widehat{\cal C}_k$, and thus $(\widehat \pi_1,\ldots,\widehat \pi_{k+1}) \in \Pi_{k+1}^*$. The desired result then follows from mathematical induction. 
	
\end{proof}

\section*{Supplement: Computation details  of (\ref{eq:opti1})}

We first replace the positive definite constraint of $\mathbf{\Omega}$ in \eqref{eq:opti1} with a slightly stronger constraint $ \mathbf{\Omega} \succeq \delta \mathbf{I}_p$ for some small $\delta>0$ \cite{ZouH2012}, and then the optimization task in \eqref{eq:opti1} becomes
\begin{align} \label{convexOpt2}
	(\widehat{\mathbf{\Theta}}, \widehat{\mathbf{\Omega}}, \widehat{\mathbf{L}}) = &\argmin_{\mathbf{\Theta}, \mathbf{\Omega}, \mathbf{L} } -\log \det \mathbf{\Theta} + \text{tr}(\mathbf{\Theta} \widehat{\mathbf{\Sigma}})  + \lambda_n ( \|\mathbf{\Omega}\|_{1,\text{off}} + \gamma \|\mathbf{L}\|_*)	\\
	& \text{such that} \ \mathbf{\Theta} = \mathbf{\Omega} + \mathbf{L}, \mathbf{\Theta} \succ 0, \mathbf{\Omega} \succeq \delta \mathbf{I}_p. \nonumber
\end{align}
The augmented Lagrangian function of \eqref{convexOpt2} is defined as 
\begin{align*}
	L_{\mu} (\mathbf{\Theta}, \mathbf{\Omega}, \mathbf{L}, \mathbf{U}) :=	&-\log \det \mathbf{\Theta} + \text{tr}(\mathbf{\Theta} \widehat{\mathbf{\Sigma}})  + \lambda_n ( \|\mathbf{\Omega}\|_{1,\text{off}} + \gamma \|\mathbf{L}\|_*) \nonumber \\
	&+ \text{tr}(\mathbf{U}(\mathbf{\Theta}-\mathbf{\Omega}-\mathbf{L})) + \frac{\mu}{2} \| \mathbf{\Theta} - \mathbf{\Omega} - \mathbf{L}\|_F^2,
\end{align*}
where $\mathbf{U}\in {\R}^{p \times p}$ is a dual variable matrix and $\mu>0$ is a penalty parameter. The corresponding dual problem is 
\begin{align} \label{dualP}
	\max_{\mathbf{U}} \min_{\mathbf{\Theta} \succ 0, \mathbf{\Omega} \succeq \delta \mathbf{I}_p, \mathbf{L}}	L_{\mu} (\mathbf{\Theta}, \mathbf{\Omega}, \mathbf{L}, \mathbf{U}).
\end{align}
It can be solved by the alternative direction method of multipliers (ADMM), consisting of the following iteration steps,
\begin{align}
	&\mathbf{\Theta}\text{-step:} \quad  \mathbf{\Theta}^{k+1} := \argmin_{\mathbf{\Theta} \succ 0} L_{\mu} (\mathbf{\Theta}, \mathbf{\Omega}^k, \mathbf{L}^k, \mathbf{U}^k), \label{admm-1}\\
	&\mathbf{\Omega}\text{-step:} \quad \mathbf{\Omega}^{k+1}:= \argmin_{\mathbf{\Omega} \succeq \delta \mathbf{I}_p} L_{\mu} (\mathbf{\Theta}^{k+1}, \mathbf{\Omega}, \mathbf{L}^k, \mathbf{U}^k), \label{admm-2}\\
	&\mathbf{L}\text{-step:} \quad  \mathbf{L}^{k+1} := \argmin_{L} L_{\mu} (\mathbf{\Theta}^{k+1}, \mathbf{\Omega}^{k+1}, \mathbf{L}, \mathbf{U}^k),  \label{admm-3} \\
	&\mathbf{U}\text{-step:} \quad \mathbf{U}^{k+1} := \mathbf{U}^k + \mu(\mathbf{\Theta}^{k+1}-\mathbf{\Omega}^{k+1}-\mathbf{L}^{k+1}).
\end{align}

For the $\mathbf{\Theta}\text{-step}$, we follow the similar procedure in \cite{Ye2011}, leading to the explicit of \eqref{admm-1},
\begin{align*}
	\mathbf{\Theta}^{k+1} = \frac{\mathbf{R}_1^k + \sqrt{(\mathbf{R}_1^k)^2+ 4\mu \mathbf{I}_p}}{2\mu},
\end{align*}
where $\mathbf{R}_1^k = \mu(\mathbf{\Omega}^k + \mathbf{L}^k)-\widehat{\mathbf{\Sigma}}-\mathbf{U}^k$ and $\sqrt{\mathbf{A}}$ denotes the square root of a symmetric positive definite matrix $\mathbf{A}$. 

For the $\mathbf{\Omega}\text{-step}$, the sub-optimization problem in \eqref{admm-2} can be re-formulated as
\begin{align} \label{admm-2-1}
	\mathbf{\Omega}^{k+1} &= \argmin_{\mathbf{\Omega} \succeq \delta \mathbf{I}_p} \lambda_n \|\mathbf{\Omega}\|_{1,\text{off}} + \text{tr}(\mathbf{U}^k (\mathbf{\Theta}^{k+1} - \mathbf{\Omega} - \mathbf{L}^k)) + \frac{\mu}{2} \| \mathbf{\Theta}^{k+1} - \mathbf{\Omega} - \mathbf{L}^k \|_F^2 \nonumber \\
	&= \argmin_{\mathbf{\Omega} \succeq \delta \mathbf{I}_p} \frac{\lambda_n}{\mu} \|\mathbf{\Omega}\|_{1,\text{off}} + \frac{1}{2} \| \mathbf{\Omega} - \mathbf{R}_2^k \|_F^2,
\end{align}
where $\mathbf{R}_2^k = \mathbf{\Theta}^{k+1} - \mathbf{L}^k + \mu^{-1} \mathbf{U}^k$. Similar as \cite{ZouH2012}, we introduce a new variable $\mathbf{\Phi} \in {\R}^{p\times p}$, and thus
\begin{align}\label{admm-2-3}
	(\mathbf{\Phi}^{k+1}, \mathbf{\Omega}^{k+1}) = &\argmin_{\mathbf{\Phi}, \mathbf{\Omega}} \frac{\lambda_n}{\mu} \|\mathbf{\Omega}\|_{1,\text{off}} + \frac{1}{2} \| \mathbf{\Omega} - \mathbf{R}_2^k \|_F^2 \\
	& \text{s.t.} \ \mathbf{\Phi} = \mathbf{\Omega}, \mathbf{\Phi} \succeq \delta \mathbf{I}_p. \nonumber
\end{align}
The augmented Lagrangian function of \eqref{admm-2-3} is defined as 
\begin{align*}
	L_{\rho}(\mathbf{\Phi}, \mathbf{\Omega}, \mathbf{V}) := \frac{\lambda_n}{\mu} \| \mathbf{\Omega}\|_{1,\text{off}} + \frac{1}{2} \| \mathbf{\Omega} - \mathbf{R}_2^k \|_F^2 + \text{tr} (\mathbf{V}(\mathbf{\Phi}-\mathbf{\Omega})) + \frac{\rho}{2} \| \mathbf{\Phi} - \mathbf{\Omega} \|_F^2,
\end{align*}
where $\mathbf{V} \in {\R}^{p\times p}$ is a dual variable matrix and $\rho>0$ is a penalty parameter. We also employ ADMM to solve \eqref{admm-2-3}, which consists of the following iterations,
\begin{align}
	&\mathbf{\Phi}^{l+1} := \argmin_{\mathbf{\Phi} \succeq \delta \mathbf{I}_p} L_{\rho} (\mathbf{\Phi}, \mathbf{\Omega}^{l}, \mathbf{V}^l),\label{2admm-1}\\
	&\mathbf{\Omega}^{l+1} := \argmin_{\mathbf{\Omega}} L_{\rho} (\mathbf{\Phi}^{l+1}, \mathbf{\Omega}, \mathbf{V}^l),\label{2admm-2}\\
	&\mathbf{V}^{l+1} := \mathbf{V}^l + \rho (\mathbf{\Phi}^{l+1} - \mathbf{\Omega}^{l+1}). \label{2admm-3}
\end{align}
Here, \eqref{2admm-1} can be solved as  
\begin{align*}
	\mathbf{\Phi}^{l+1} &= \argmin_{\mathbf{\Phi} \succeq \delta \mathbf{I}_p} \text{tr}(\mathbf{V}^l (\mathbf{\Phi}-\mathbf{\Omega}^l)) + \frac{\rho}{2} \| \mathbf{\Phi} - \mathbf{\Omega}^l \|_F^2 \\
	& = \argmin_{\mathbf{\Phi} \succeq \delta \mathbf{I}_p} \| \mathbf{\Phi} - (\mathbf{\Omega}^l - \rho^{-1}\mathbf{V}^l ) \|_F^2 = (\mathbf{\Omega}^l - \rho^{-1}\mathbf{V}^l)_+,
\end{align*}
where $(\mathbf{A})_+$ denotes the projection of a matrix $\mathbf{A}$ onto the convex cone $\{ \mathbf{\Phi} \succeq \delta \mathbf{I}_p\}$. Specifically, if the eigen decomposition of $\mathbf{A}$ is $\sum_{j=1}^p s_j \mathbf{v}_j \mathbf{v}_j^T$, then $(\mathbf{A})_+ = \sum_{j=1}^p \max(s_j, \delta) \mathbf{v}_j \mathbf{v}_j^T$. Furthermore, \eqref{2admm-2} can be solved as
\begin{align*}
	\mathbf{\Omega}^{l+1} &= \argmin_{\mathbf{\Omega}} \frac{\lambda_n}{\mu} \| \mathbf{\Omega}\|_{1,\text{off}} + \frac{1}{2} \| \mathbf{\Omega} - \mathbf{R}_2^k \|_F^2 + \text{tr} (\mathbf{V}(\mathbf{\Phi}^{l+1}-\mathbf{\Omega})) + \frac{\rho}{2} \| \mathbf{\Phi}^{l+1} - \mathbf{\Omega} \|_F^2\\
	&= 	\argmin_{\mathbf{\Omega}} \frac{\lambda_n}{\mu (1+\rho)}  \| \mathbf{\Omega} \|_{1,\text{off}} + \frac{1}{2}\| \mathbf{\Omega} - \frac{\mathbf{R}_2^k + \mathbf{V}^l + \rho \mathbf{\Phi}^{l+1}}{1+\rho} \|_F^2\\
	&= \frac{1}{1+\rho} {\cal S}_{\frac{\lambda_n}{\mu}} (\mathbf{R}_2^k + \mathbf{V}^l + \rho \mathbf{\Phi}^{l+1}),
\end{align*}
where ${\cal S}_{\tau}(\mathbf{A}) = (s_{\tau}(A_{ij}))_{ij} \in {\R}^{p \times p}$, and $s_{\tau}(A_{ij})=\sign(A_{ij})\max(|A_{ij}|-\tau,0)\mathbf{1}(i\neq j) + A_{ij} \mathbf{1}(i=j)$ denotes entry-wise soft-thresholding operator for off-diagonal entries of $\mathbf{A}$. 

For the $\mathbf{L}\text{-step}$, it follows from Theorem 2.1 in \cite{Cai2010} that
\begin{align*}
	\mathbf{L}^{k+1} &= \argmin_{\mathbf{L}} \lambda_n \gamma \| \mathbf{L} \|_* + \text{tr} (\mathbf{U}^k (\mathbf{\Theta}^{k+1} - \mathbf{\Omega}^{k+1}-\mathbf{L})) + \frac{\mu}{2} \|  \mathbf{\Theta}^{k+1} - \mathbf{\Omega}^{k+1} - \mathbf{L} \|_F^2 \\
	&= \argmin_{\mathbf{L}} \frac{\lambda_n \gamma}{\mu} \| \mathbf{L} \|_* + \frac{1}{2} \| \mathbf{L} - (\mathbf{\Theta}^{k+1} - \mathbf{\Omega}^{k+1} + \mu^{-1} \mathbf{U}^k) \|_F^2 \\
	&= {\cal T}_{\frac{\lambda_n \gamma}{\mu}} (\mathbf{\Theta}^{k+1} - \mathbf{\Omega}^{k+1} + \mu^{-1} \mathbf{U}^k),
\end{align*}
where ${\cal T}_{\tau}(\mathbf{A})$ denotes the singular value shrinkage operator for a matrix $\mathbf{A}$. Specifically, if the SVD of $\mathbf{A}$ with rank $r$ is $\sum_{j=1}^p d_j \mathbf{u}_j \mathbf{v}_j^T$, then ${\cal T}_{\tau}(\mathbf{A}) = \sum_{j=1}^r \max(d_j-\tau,0)\mathbf{u}_j \mathbf{v}_j^T$.

\bibliography{ref} 
\bibliographystyle{apalike}

\end{document}